\documentclass[12pt]{article}

\usepackage[table]{xcolor}

\usepackage[a4paper, left=2.5cm, right=2.5cm, top=2.5cm, bottom=3cm]{geometry}
\usepackage{parskip}
\usepackage[utf8]{inputenc}

\usepackage[]{placeins} 

\newcommand\ind{\protect\mathpalette{\protect\independenT}{\perp}} 
\def\independenT#1#2{\mathrel{\rlap{$#1#2$}\mkern4mu{#1#2}}}
\newcommand{\nind}{\not\ind}

\usepackage{tikz}
\usetikzlibrary{arrows}

\usepackage{amssymb}
\usepackage{amsthm}
\usepackage{amsmath}
\usepackage{blkarray} 
\usepackage{bm}

\usepackage{booktabs}
\usepackage{longtable}
\usepackage{caption}

\usepackage[authoryear]{natbib}
\usepackage[hidelinks]{hyperref}
\usepackage[numbered]{bookmark}

\newtheoremstyle{newline}
{\parskip}
{\topsep}
{\itshape}
{}
{\bfseries}
{}
{\newline}
{}
\theoremstyle{newline}

\newtheorem{theorem}{Theorem}
\newtheorem{proposition}[theorem]{Proposition}
\newtheorem{lemma}[theorem]{Lemma}
\newtheorem{example}[theorem]{Illustration}
\newtheorem{definition}[theorem]{Definition}

\begin{document}
	
\begin{center}
~

{\LARGE \bfseries Multiple imputation and test-wise deletion}

{\LARGE \bfseries for causal discovery with incomplete}

{\LARGE \bfseries cohort data}

~
\vspace{0.5cm}

\today	
\end{center}

\vspace{1.5cm}

\textbf{Janine Witte}\textsuperscript{1,2}, \textbf{Ronja Foraita}\textsuperscript{1}, \textbf{Vanessa Didelez}\textsuperscript{1,2}\\

\textsuperscript{1} Leibniz Institute for Prevention Research and Epidemiology---BIPS\\
\textsuperscript{2} University of Bremen\\

\vspace{2cm}

{\large \bfseries ABSTRACT}\smallskip

Causal discovery algorithms estimate causal graphs from observational data. This can provide a valuable complement to analyses focussing on the causal relation between individual treatment-outcome pairs. Constraint-based causal discovery algorithms rely on conditional independence testing when building the graph. Until recently, these algorithms have been unable to handle missing values. In this paper, we investigate two alternative solutions: Test-wise deletion and multiple imputation. We establish necessary and sufficient conditions for the recoverability of causal structures under test-wise deletion, and argue that multiple imputation is more challenging in the context of causal discovery than for estimation. We conduct an extensive comparison by simulating from benchmark causal graphs: As one might expect, we find that test-wise deletion and multiple imputation both clearly outperform list-wise deletion and single imputation. Crucially, our results further suggest that multiple imputation is especially useful in settings with a small number of either Gaussian or discrete variables, but when the dataset contains a mix of both neither method is uniformly best. The methods we compare include random forest imputation and a hybrid procedure combining test-wise deletion and multiple imputation. An application to data from the IDEFICS cohort study on diet- and lifestyle-related diseases in European children serves as an illustrating example.
\vfill

\textbf{Keywords:} causal search, causal inference, MICE, missing values, PC-algorithm, structure learning


\newpage
\section{Introduction}

Causal graphs have become very popular in epidemiology and other disciplines as a means to represent the causal structure among random variables \citep{GreenlandPearlRobins1999, Tennantetal2021, MorganWinship2014, Cunningham2021}. A causal graph drawn based on background knowledge helps communicating causal assumptions, and can guide variable selection when estimating a causal effect \citep{Didelez2018}. In contrast, the aim of \textit{causal discovery} is to infer a plausible graph or set of graphs from data when the causal structure is not known a priori. The estimated graphs can be used to support or challenge existing theories, to generate new hypotheses, or to estimate possible causal effects consistent with the data \citep{MaathuisKalischBuhlmann2009}. Since its introduction in the 1980s, causal discovery has been applied in a variety of fields including epidemiology \citep{Moffaetal2017}, medical imaging \citep{Rayetal2015}, genome-wide association studies \citep{Alekseyenkoetal2011}, education research \citep{RauScheines2012}, stock market research \citep{BesslerYang2003}, linguistics \citep{RobertsWinters2013} and climate research \citep{EbertUphoffDeng2012}.

Popular causal discovery methods are constraint-based algorithms, which search for conditional independencies between the variables and reconstruct the causal structure so as to satisfy the constraints imposed by these independencies. A main advantage of the constraint-based approach is its flexibility. As the algorithms mainly rely on conditional independence testing, they can be applied to any type of data (continuous, categorical, ordinal, mixed etc.), as long as suitable tests are available. Moreover, constraint-based algorithms can in principle be applied even in the presence of latent variables \citep{SpirtesGlymourScheines2000, Zhang2008}.

Most software implementations of constraint-based causal discovery require fully observed data as an input. Simple ways of dealing with incomplete data lead to unsatisfactory results: Under list-wise deletion, also called complete-case analysis, all incomplete records are deleted, which can severely reduce the sample size and induce selection bias. Single imputation usually leads to underestimation of standard errors. Recently, two promising new strategies have been suggested for constraint-based causal discovery with missing values: (i) test-wise deletion \citep{StroblVisweswaranSpirtes2018, Tuetal2019, Tuetal2020}, where each conditional independence test is performed using the subset of records containing complete data for all variables involved in that particular test, and (ii) multiple imputation for Gaussian data \citep{Foraitaetal2020}.

In this paper, we formally investigate, generalise and compare test-wise deletion and multiple imputation in the context of causal discovery. Building on \cite{Tuetal2019}, we establish necessary and sufficient conditions for the recoverability of causal graphs under test-wise deletion. Further, we extend the multiple imputation approach by \cite{Foraitaetal2020} to discrete and mixed variables, characterise situations in which multiple imputation is expected to outperform test-wise deletion, and discuss why selecting the imputation model is challenging in causal discovery. The performance of list-wise deletion, test-wise deletion, single imputation and multiple imputation is compared on simulated and real data. Our findings are not only useful for causal discovery; they also provide insights into the general problem of conditional independence testing with missing values, e.g.\ necessary and sufficient conditions for identification of (in)dependencies.

\subsection{Motivating example: the IDEFICS study}
\label{sec:motivation}

Our work was motivated by IDEFICS (Identification and prevention of dietary and lifestyle-induced health effects in children and infants study), a prospective cohort study including 16\,229 children from eight European countries. The children were first examined in 2007/2008, and a follow-up examination took place two years later. The cohort was later extended by the I.Family study \citep{Ahrensetal2017}.

Designed to identify factors relating to childhood obesity and other non-communicable health conditions, the IDEFICS study included measurements on diet, lifestyle, living environment, socio-economic background and mental and physical health. Even though these factors are known to interact in a complex manner \citep{Leeetal2017, VandenbroeckGoossensClemens2017}, analyses of the IDEFICS data often focus on individual exposures and/or individual outcomes (e.g.\ \citealp{Bornhorstetal2016, Hebestreitetal2016, Pohlabelnetal2017}). A causal discovery analysis would therefore be a valuable addition to the analyses conducted so far.

\begin{figure}[h]
	\begin{center}
		\includegraphics[trim=0 2mm 0 6mm, clip, width=\textwidth]{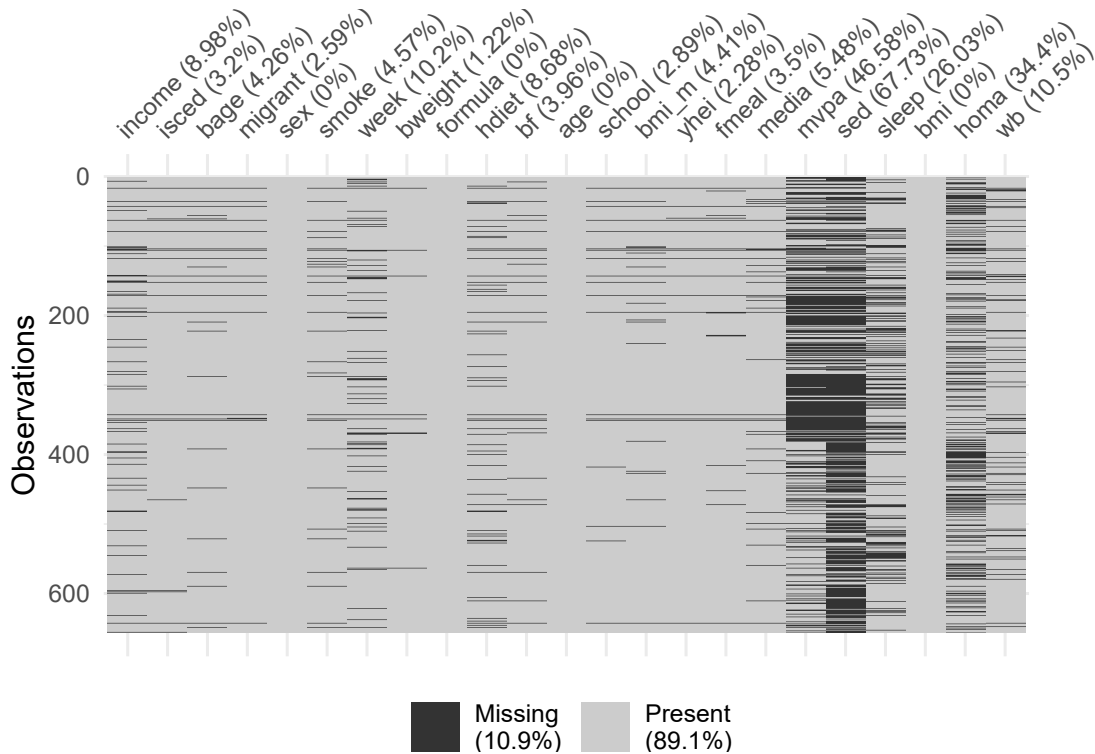}
		\caption{Missingness pattern in selected IDEFICS variables. The numbers in parentheses indicate the missingness percentage per variable.}
		\label{fig:vismiss}
	\end{center}
\end{figure}

Like most observational datasets, the IDEFICS data contain missing values. Figure~\ref{fig:vismiss} visualises the missingness pattern in a subsample of the IDEFICS data containing 657 children from Germany. The choice of variables roughly follows \cite{Foraitaetal2021}, who performed causal discovery on a larger subsample of the IDEFICS and I.Family data including more children and time points. See Table~\ref{tab:variables} for details on the variables used in the present paper. The overall proportion of missing data points is 10.9\,\%. Only 78 data rows (11.9\,\%) are completely observed, hence list-wise deletion would reduce the sample size by almost 90\,\%. It is therefore clear that a more efficient method for dealing with the missing values is needed.

\subsection{Outline}

The paper is organised as follows: We start with a brief review of causal graphs and constraint-based causal discovery in Section~\ref{sec:causaldiscovery}. In Section~\ref{sec:mechanisms}, we contrast Rubin's classification of missingness with the newer concept of the missingness graph. Section~\ref{sec:missing} contains the theoretical results on test-wise deletion and multiple imputation, and a comparison of their performance on data simulated using simple graphical structures. A comprehensive simulation study for benchmark settings is described in Section~\ref{sec:simulation}. Section~\ref{sec:realdata} contains an application to the IDEFICS data. We conclude with a discussion in Section~\ref{sec:conclusions}. All new methods are implemented in the \texttt{R} package \texttt{micd} available on GitHub (\url{www.github.com/bips-hb/micd}).
\vspace{3cm}

\FloatBarrier

\begin{table}[h]
	\caption{Baseline variables of IDEFICS. The data were log-transformed as indicated in order to reduce skewness of the marginal distributions.}
	\label{tab:variables}
	\begin{tabular}{lp{13.5cm}}
		\toprule
		income & Household income (three categories) \\
		isced & Parent's education (three categories) \\
		bage & Mother's age in years when the child was born (continuous) \\
		migrant & Migration status of child (binary) \\
		sex & Sex of the child (binary) \\
		smoke & Mother smoked during pregnancy (binary) \\
		week & Completed weeks of pregnancy (continuous) \\
		bweight & Birthweight in g (continuous) \\
		formula & Child received formula milk (binary) \\
		hdiet & Months until child was integrated into household diet (continuous, log-transformed)\\
		bf & Total duration of breastfeeding in months (continuous, log-transformed) \\
		age & Age in years of the child upon inclusion in the study (continuous) \\
		school & Child visits kindergarten or school (three categories) \\
		bmi\_m & Mother's BMI (continuous) \\
		yhei & Child's youth healthy eating score (continuous) \\
		fmeal & Child eats breakfast at home 7 days a week (binary) \\
		media & Child's audiovisual media consumption in hours/day (continuous)\\
		mvpa & Child's physical activity in hours/week (continuous, log-transformed)\\
		sed & Child's sedentary behaviour in hours/week (continuous)\\
		sleep & Child's sleep duration in hours (continuous)\\
		bmi & Child's BMI z-score (continuous) \\
		homa & Child's HOMA insulin resistance index (continuous)\\
		wb & Child's well-being score (continuous)\\
		\bottomrule
	\end{tabular}
\end{table}

\FloatBarrier

\section{Background on causal discovery}
\label{sec:causaldiscovery}

In this section, we review causal discovery with complete data.

\subsection{(Causal) graphs}

We start by defining the required graphical terminology.

\textbf{Nodes, edges and cycles.} A \textit{graph} consists of a set of nodes \(\mathbf{V}\) and a set of edges \(\mathbf{E}\subseteq\mathbf{V}\times\mathbf{V}\). Here, graphs have at most one edge between a given pair of nodes, and edges are either directed (\(\rightarrow\)) or undirected (\(-\)). An edge from a node to itself is not allowed. Two nodes connected by an edge are \textit{adjacent}. If \(V_i\rightarrow V_j\), then \(V_i\) is a \textit{parent} of \(V_j\) and \(V_j\) is a \textit{child} of \(V_i\). If \(V_i - V_j\), then \(V_i\) and \(V_j\) are \textit{neighbours}. A sequence of nodes \((V_1,\dots,V_P)\) with \(V_1=V_P\) such that for \(1\le i<P\), there is a directed edge \(V_i\rightarrow V_{i+1}\), is called a \textit{directed cycle}. A \textit{directed acyclic graph} (\textit{DAG}) is a graph with only directed edges and without directed cycles. The \textit{skeleton} of a DAG \(\mathcal{D}\) has the same nodes and adjacencies as \(\mathcal{D}\), but only undirected edges.

\textbf{Paths.} A sequence of distinct nodes \((V_1,\dots,V_P)\) such that for \(1\le i<P\), \(V_i\) and \(V_{i+1}\) are adjacent, is called a \textit{path} between \(V_1\) and \(V_P\). If in addition for \(1\le i<P\), \(V_i\rightarrow V_{i+1}\), then the path is \textit{directed} from \(V_1\) to \(V_P\). A node \(V_i\) is a \textit{descendant} of a node \(V_j\) if either \(V_i=V_j\) or there is a directed path from \(V_j\) to \(V_i\).

\textbf{Colliders and d-separation.} Consider a path \(p=(V_1,\dots,V_P)\) in a DAG \(\mathcal{D}\) with node set \(\mathbf{V}\). For \(1< i<P\), the node \(V_i\) is a \textit{collider} on \(p\) if \(V_{i-1}\rightarrow V_i\leftarrow V_{i+1}\); otherwise, \(V_i\) is a \textit{non-collider} on \(p\). The path \(p\) is \textit{open} given a set of nodes \(\mathbf{Z}\subseteq\mathbf{V}\) if (i) no non-collider on \(p\) is in \(\mathbf{Z}\) and (ii) every collider on \(p\) has a descendant in \(\mathbf{Z}\). Otherwise, \(p\) is \textit{blocked} given \(\mathbf{Z}\). For disjoint sets of nodes \(\mathbf{X},\mathbf{Y},\mathbf{Z}\subset\mathbf{V}\), \(\mathbf{X}\) and \(\mathbf{Y}\) are \textit{d-separated} by \(\mathbf{Z}\) in \(\mathcal{D}\) if every path between a node \(X\in\mathbf{X}\) and a node \(Y\in\mathbf{Y}\) is blocked given \(\mathbf{Z}\). This is denoted as \(\mathbf{X}\perp_\mathcal{D}\mathbf{Y}\mid\mathbf{Z}\).

Consider a set of random variables \(\mathbf{V}=\{V_1,\dots,V_K\}\), which can be continuous or discrete or a mix thereof. We assume that the causal structure among the variables in \(\mathbf{V}\) can be represented by a \textit{causal DAG} \(\mathcal{D}\) with node set \(\mathbf{V}\). In particular, we assume that the joint density \(f(\mathbf{v})=f(v_1,\dots,v_K)\) of \(\mathbf{V}\) is \textit{Markov} and \textit{faithful} to \(\mathcal{D}\). The Markov assumption requires that \(f(\mathbf{v})\) factorises as \(f(\mathbf{v})=\prod_{k=1}^K f(v_k\mid \mathrm{pa}(V_k,\mathcal{D}))\), where \(\mathrm{pa}(V_k,\mathcal{D})\) denotes the set of parents of the node \(V_k\) in \(\mathcal{D}\). Under this assumption, every d-separation \(\mathbf{X}\perp_\mathcal{D}\mathbf{Y}\mid\mathbf{Z}\) in the graph corresponds to a conditional independence \(\mathbf{X}\ind\mathbf{Y}\mid\mathbf{Z}\) in the distribution, where the latter is read as `\(\mathbf{X}\) and \(\mathbf{Y}\) are conditionally independent given \(\mathbf{Z}\)'. The faithfulness assumption requires the reverse to be true as well, i.e.\ every conditional independence in the distribution corresponds to a d-separation in the graph.

In order to give the DAG \(\mathcal{D}\) a causal interpretation, we additionally assume that if we intervened in the physical system underlying the random variables and fixed the value of a variable \(V_j\) in \(\mathcal{D}\) to \(v_j\), then the resulting distribution of the remaining variables would still factorise as \(f(v_1,\dots,v_{j-1},v_{j+1},\dots,v_K)=\prod_{k\in\{1,\dots,K\}\setminus j} f(v_k\mid \mathrm{pa}(V_k,\mathcal{D}))\). This is plausible only if there are no latent variables, i.e.\ variables not in the graph, representing common causes of two or more variables in the graph. Hence, we assume the absence of such variables. This assumption in particular is known as \textit{causal sufficiency}.

\subsection{Causal discovery}

The core idea of constraint-based causal discovery is to search for conditional independencies in the data, and use them to reconstruct the graph. However, as several DAGs can imply the same set of d-separations and hence conditional independencies, it is not possible in general to infer a single DAG from observational data alone, even if all above assumptions hold and the sample is infinitely large. The set of all DAGs implying a given set of d-separations is called a \textit{(Markov) equivalence class} and can uniquely be represented by a so-called \textit{completed partially directed acyclic graph} (\textit{CPDAG}) with directed and undirected edges. An undirected edge in a CPDAG means that both orientations occur within the equivalence class. Without further background knowledge or parametric assumptions, constraint-based causal discovery can at best recover the true CPDAG.

In this paper, we consider the most popular constraint-based causal discovery algorithm, which is the PC-algorithm\footnote{PC was named after its inventors, Peter Spirtes and Clark Glymour.} \citep{SpirtesGlymourScheines2000}. PC starts with a fully connected undirected graph and proceeds in three steps. First, a series of conditional independence tests is performed for each pair of variables \((X,Y)\). If \(X\) and \(Y\) are found to be conditionally independent for some conditioning set, the edge between them is deleted. In order to keep the number of performed tests small, the conditioning sets are always chosen from among the nodes adjacent to \(X\) or the nodes adjacent to \(Y\) in the current graph. The resulting undirected graph is the estimated skeleton. Second, PC searches for triples of variables \((X,Y,Z)\) such that (i) the estimated skeleton contains a path \(X-Y-Z\), (ii) \(X\) and \(Z\) are not adjacent in the estimated skeleton, and (iii) \(X\) and \(Z\) are conditionally independent given a set (or several sets, see \citealp{ColomboMaathuis2014}) of variables not containing \(Y\). The path is then oriented as \(X\rightarrow Y\leftarrow Z\). Third, additional edges are oriented according to logical rules \citep{Meek1995}. It can be shown that PC recovers the true CPDAG if the above assumptions of faithfulness and causal sufficiency hold and correct conditional (in)dependence information is provided \citep{SpirtesGlymourScheines2000}. If background knowledge is available, e.g.\ in the form of a partial node ordering, the output of PC can be a graph with more directed edges than the CPDAG \citep{Meek1995}.

\subsection{Conditional independence testing}

Conditional independence tests commonly used for the PC-algorithm are Fisher's \(z\)-test for continuous data and the \(G^2\)-test for categorical data. Briefly, Fisher's \(z\)-test tests for a zero conditional correlation, assuming that the variables in the test follow a multivariate Gaussian distribution. The \(G^2\)-test is a non-parametric conditional independence test for contingency tables. It can also be viewed as a likelihood-ratio test under a saturated multinomial model. If a dataset contains both continuous and categorical data, common strategies are to either discretise the continuous variables, or to treat the categorical variables as continuous. For the case that the variables in the test jointly follow a Conditional Gaussian (CG) distribution \citep{LauritzenWermuth1989}, \cite{AndrewsRamseyCooper2018} described a likelihood-ratio test, which we call the `CG-test'. More details on Fisher's \(z\)-test, the \(G^2\)-test and the CG-test are given in Appendix~\ref{app:condindtests}.

The significance level \(\alpha\) for the conditional independence tests performed within the PC-algorithm has the role of a tuning parameter, where a smaller value leads to a sparser~graph.

\section{Missingness mechanisms and missingness graphs}
\label{sec:mechanisms}

Assume that a subset \(\mathbf{V}^*\subseteq\mathbf{V}\) of the variables may contain missing values. For each \(V\in\mathbf{V}^*\), we define a response indicator \(R_V\) that is 1 if \(V\) is observed, and 0 if \(V\) is missing. The response indicators are themselves binary random variables. We denote as \(\mathbf{R}(\mathbf{V})=\{R_V:V\in\mathbf{V}^*\}\) the set of all variable-wise response indicators. Further, for a subset \(\mathbf{A}\subseteq\mathbf{V}\), we define \(R^\mathbf{A}\) to be 1 if all variables in \(\mathbf{A}\) are observed, and 0 otherwise.

In line with the literature, we assume that the missing values are not known, but exist. We refer to the distribution of the variables had all values been measured as the \textit{full-data distribution}.

Next, we discuss two ways of describing the relation between the substantive variables \(\mathbf{V}\) and the response indicators \(\mathbf{R}(\mathbf{V})\), i.e.\ the \textit{missingness mechanism}. The traditional classification according to \cite{Rubin1976} (see Section~\ref{sec:Rubin}) is relevant for multiple imputation, which requires the data to be missing at random. In more recent work (e.g.\ \citealp{MohanPearl2021}), assumptions about the missingness mechanism are encoded in a causal graph over \(\mathbf{V}\cup\mathbf{R}(\mathbf{V})\) (see Section~\ref{sec:Mohan}). Due to its graphical nature, this alternative framework combines well with the concept of causal discovery. In particular, it can be used to assess the identifiability of conditional (in)dependencies under test-wise deletion, see \cite{Tuetal2019} and Section~\ref{sec:twd} below.

\subsection{Rubin's classification of missingness}
\label{sec:Rubin}

Three classes of missingness mechanisms are often distinguished in the literature \citep{Rubin1976}: Values are said to be \textit{missing completely at random} (MCAR) 
if \(f(\mathbf{r}\mid \mathbf{v})=f(\mathbf{r})\), i.e.\ missingness is independent of the substantive variables. Values are said to be \textit{missing at random} (MAR) if, for each individual \(i\) in the dataset, \(f(\mathbf{r}_i\mid \mathbf{v}_i)=f(\mathbf{r}_i\mid \mathbf{v}^O_i)\), where \(\mathbf{V}^O_i\) is the set of variables that is observed for individual \(i\). MAR thus expresses that for each individual, missingness may be associated with the observed variables, but is conditionally independent of the unobserved variables. If values are not MCAR or MAR, they are said to be \textit{missing not at random} (MNAR). Whether values in a given dataset are MAR cannot be determined empirically. Note that the conditioning set \(\mathbf{V}^O_i\) in the MAR equation may contain different variables for each individual, hence the equation corresponds to conditional independence between `events', not between random variables \citep{Seamanetal2013, MealliRubin2015, DorettiGenelettiStanghellini2018}. This can make the MAR assumption difficult to justify in practice. As an example, consider two incompletely observed variables \textit{BMI} and \textit{well-being}, where \textit{BMI} is MAR given \textit{well-being}. This implies that the missingness of \textit{BMI} may depend on the value of \textit{well-being} only in those individuals for whom \textit{well-being} is observed, while for the other individuals, missingness of \textit{BMI} and \textit{well-being} must be independent.

Rubin's categories of missingness mechanisms were derived in the context of likelihood inference. For instance, under MAR, regression parameters and their standard errors can consistently be estimated in the presence of missing data using multiple imputation as discussed below.

\subsection{Missingness graphs}
\label{sec:Mohan}

A recent line of work uses \textit{missingness graphs} to encode assumptions about the missingness mechanism \citep{Danieletal2012, Westreich2012, MohanPearlTian2013, MorenoBetancuretal2018, MohanPearl2021}. These graphs include both the substantive variables \(\mathbf{V}\) as well as the response indicators \(\mathbf{R}(\mathbf{V})\) as nodes, where it is assumed that the response indicators do not cause the substantive variables (i.e.\ there are no directed edges from nodes in \(\mathbf{R}(\mathbf{V})\) to nodes in \(\mathbf{V}\)). A set-wise missingness indicator \(R^{\mathbf{A}}\) as defined above is represented as a child of all nodes in \(\{R_A:A\in\mathbf{A}\}\). The usual rules of d-separation can then be used to determine whether aspects of the full-data distribution are identified from the observed data. In this paper, we only consider missingness graphs that are DAGs and call them \textit{missingness DAGs}.

Consider the missingness DAG in Figure~\ref{fig:mgraph} as an example. It shows three substantive variables, \textit{age} (\textit{A}), \textit{physical activity} (\textit{P}) and \textit{systolic blood pressure} (\textit{S}), together with their response indicators \(R_P\) and \(R_S\). As \textit{age} is assumed to be fully observed, its response indicator is omitted. According to this graph, the missingness of \textit{physical activity} depends on \textit{age} and \textit{systolic blood pressure}. Using the rules of do-calculation as described in \cite{MohanPearlTian2013}, it can be established e.g.\ that the full-data joint density \(f(a,p,s)\) of the substantive variables can be identified from the incompletely observed variables as \(f(a,p,s)=f(p\mid a,s,R_P=1,R_S=1)f(s\mid a,R_S=1)f(a)\). Note that under a causal interpretation, the graph depicts the assumption that the nodes in the graph (including the response indicators) do not share common causes except where shown in the graph; for example, we assume that \textit{age} is the only common cause of \textit{physical activity} and \textit{systolic blood pressure}.

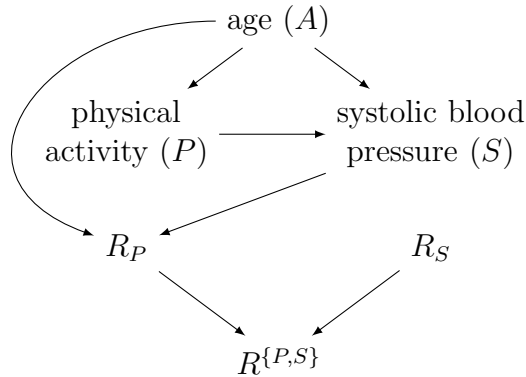
\begin{figure}
		\begin{center}
		\begin{tikzpicture}[node distance=15mm, >=latex]
		\node (A) {age (\(A\))};
		\node[below of=A, xshift=-20mm, align=center] (P) {physical\\ activity (\(P\))};
		\node[below of=A, xshift=20mm, align=center] (S) {systolic blood\\ pressure (\(S\))};
		\node[below of=P] (RP) {\(R_P\)};
		\node[below of=S] (RS) {\(R_S\)};
		\node[below of=RP, xshift=20mm] (RPS) {\(R^{\{P,S\}}\)};
		\draw[->] (A) to (P);	
		\draw[->] (A) to (S);
		\draw[->] (P) to (S);
		\draw[->] (S) to (RP);
		\draw[->, out=180, in=160, looseness=1.8] (A.west) to (RP);
		\draw[->] (RP) to (RPS);
		\draw[->] (RS) to (RPS);		
		
		\end{tikzpicture}
	\end{center}
	\caption{Example missingness DAG. }
	\label{fig:mgraph}
\end{figure}

Missingness graphs can only represent dependence relations between variables, not between events. Therefore, Rubin's MAR assumption cannot be depicted in a missingness graph. For example, it cannot be determined from the graph in Figure~\ref{fig:mgraph} whether \textit{physical activity} is MAR or MNAR, since its missingness could depend, in some individuals, on \textit{systolic blood pressure} values that are themselves missing. \cite{MohanPearlTian2013} proposed an alternative, variable-based definition of MAR, which is, however, not immediately relevant for the present paper.

\section{Test-wise deletion and multiple imputation for \\constraint-based causal discovery}
\label{sec:missing}

In this section, we investigate the assumptions under which conditional (in)dependencies are identified under multiple imputation or test-wise deletion, and discuss how different aspects affect the power of the conditional independence tests. As we will see, the answers are not necessarily the same as for the estimation of regression coefficients, which has been the primary focus of missing data methods.

\subsection{Test-wise deletion}
\label{sec:twd}

Consider \(X\in\mathbf{V}\), \(Y\in\mathbf{V}\setminus\{X\}\) and \(\mathbf{Z}\subset\mathbf{V}\setminus\{X,Y\}\). Test-wise deletion means that the conditional independence \(X\ind Y\mid\mathbf{Z}\) is tested in the subsample of the data where \(X\), \(Y\) and \(\mathbf{Z}\) are fully observed (irrespective of missing values in other variables). Formally, this implies testing \(X\ind Y\mid(\mathbf{Z},R^{XY\mathbf{Z}}=1)\), where we defined \(R^{XY\mathbf{Z}}=R^{\{X,Y\}\cup\mathbf{Z}}\) for better readability. We say that a \textit{conditional independence} in the full-data distribution is \textit{identified under test-wise deletion} if
\[X\ind Y\mid\mathbf{Z} \Rightarrow X\ind Y\mid(\mathbf{Z},R^{XY\mathbf{Z}}=1).\]
Vice versa, we say that a \textit{conditional dependence is identified under test-wise deletion} if
\[X\nind Y\mid\mathbf{Z} \Rightarrow X\nind Y\mid(\mathbf{Z},R^{XY\mathbf{Z}}=1).\]

Assume that the distribution of \(\mathbf{V}\cup\mathbf{R}(\mathbf{V})\) is faithful to a missingness DAG. \cite{Tuetal2019} showed that in this setting, conditional dependencies (but not independencies) are identified under test-wise deletion under an additional assumption they termed \textit{faithful observability}:
\[X\ind Y\mid(\mathbf{Z},R^{XY\mathbf{Z}}=1) \Leftrightarrow X\ind Y\mid(\mathbf{Z},R^{XY\mathbf{Z}}=0).\]
In words, an independence in the distribution underlying the data used in the test must also be present in the distribution underlying the (partially) unobserved data not used in the test. We further show in Appendix~\ref{app:proofs} that under faithful observability, a conditional independence \(X\ind Y\mid\mathbf{Z}\) is identified under test-wise deletion if and only if \(R^{XY\mathbf{Z}}\ind X\mid(Y,\mathbf{Z})\) or \(R^{XY\mathbf{Z}}\ind Y\mid(X,\mathbf{Z})\). 
Based on these results, we next formulate a necessary and sufficient condition for the validity of the PC-algorithm using as input correct information about the (in)dependencies in the distributions under test-wise deletion (`oracle test-wise-deletion PC'). We use \(\mathrm{adj}(V,\mathcal{D})\) to denote the set of nodes adjacent to node \(V\) in DAG \(\mathcal{D}\).

\begin{definition}[Admissible separator condition]
	Let \(\mathcal{D}\) be a missingness DAG with node set \(\mathbf{V}\cup\mathbf{R}(\mathbf{V})\). We say that the \emph{admissible separator condition} holds if for all pairs \((X,Y)\) of non-adjacent nodes in \(\mathbf{V}\), there exists a (possibly empty) set \(\mathbf{Z}\subset\mathbf{V}\) such that (i) \(X\ind Y\mid\mathbf{Z}\), (ii) \(\mathbf{Z}\subseteq\mathrm{adj}(X,\mathcal{D})\) or \(\mathbf{Z}\subseteq\mathrm{adj}(Y,\mathcal{D})\) and (iii) \(R^{XY\mathbf{Z}}\ind X\mid(Y,\mathbf{Z})\) or \(R^{XY\mathbf{Z}}\ind Y\mid(X,\mathbf{Z})\).
\end{definition}

\newpage
\begin{proposition}
		\label{prop:twd}
		Let \(\mathcal{D}\) be a missingness DAG with node set \(\mathbf{V}\cup\mathbf{R}(\mathbf{V})\), such that the distribution of \(\mathbf{V}\cup\mathbf{R}(\mathbf{V})\) is faithful to \(\mathcal{D}\), and assume that faithful observability holds. Then oracle test-wise-deletion PC recovers the true CPDAG over \(\mathbf{V}\) if and only if the admissible separator condition holds. 
\end{proposition}

A proof of Proposition~\ref{prop:twd} is given in Appendix~\ref{app:proofs}. If faithful observability holds but the admissible separator condition does not hold, then the discovered CPDAG has additional edges compared to the true CPDAG, and may contain erroneous edge orientations. The admissible separator condition is not empirically verifiable and arguably difficult to assess in practice, where the true graph is not known. Consider the four missingness DAGs in Figure~\ref{fig:identification} for illustration. The missingness structure is the same in all graphs (i.e.\ \(Y\) is missing depending on the values of \(X\) and \(Y\) itself), but whether the CPDAG is correctly discovered by oracle test-wise-deletion PC under the assumptions of Proposition~\ref{prop:twd}, depends on the presence or absence of the edge \(X-Y\). Note that the correct CPDAG is recovered for the DAGs 1) and 2) in Figure~\ref{fig:identification} even though \(C\) is MNAR (as the missingness of \(Y\) depends on the values of \(Y\) itself). Consider also the missingness DAG in Figure~\ref{fig:MAR*}. Here the missingness depends on fully observed variables only, which implies that the MAR assumption holds. The conditional independence \(X\ind Y\mid Z\) is not identified under test-wise deletion, however, as neither \(R_Z\ind X\mid (Y,Z)\) nor \(R_Z\ind Y\mid (X,Z)\).

\begin{figure}[hb]
	\vspace{5mm}
	\begin{center}
	\begin{tikzpicture}[node distance=20mm, >=latex]
	\node (A1) {\(X\)};
	\node[above of=A1, xshift=10mm, yshift=-5mm] (B1) {\(Z\)};
	\node[right of=A1] (C1) {\(Y\)};
	\node[below of=B1, yshift=-10mm] (R1) {\(R_Y\)};
	\draw[->] (A1) to (B1);	
	\draw[->] (A1) to (C1);
	\draw[->] (B1) to (C1);
	\draw[->] (A1) to (R1);
	\draw[->] (C1) to (R1);
	
	\node[right of=C1] (A2) {\(X\)};
	\node[above of=A2, xshift=10mm, yshift=-5mm] (B2) {\(Z\)};
	\node[right of=A2] (C2) {\(Y\)};
	\node[below of=B2, yshift=-10mm] (R2) {\(R_Y\)};
	\draw[->] (C2) to (A2);
	\draw[->] (A2) to (B2);
	\draw[->] (A2) to (R2);
	\draw[->] (C2) to (R2);
	
	\node[right of=C2] (A3) {\(X\)};
	\node[above of=A3, xshift=10mm, yshift=-5mm] (B3) {\(Z\)};
	\node[right of=A3] (C3) {\(Y\)};
	\node[below of=B3, yshift=-10mm] (R3) {\(R_Y\)};
	\draw[->] (B3) to (A3);	
	\draw[->] (B3) to (C3);
	\draw[->] (A3) to (R3);
	\draw[->] (C3) to (R3);
	
	\node[right of=C3] (A4) {\(X\)};
	\node[above of=A4, xshift=10mm, yshift=-5mm] (B4) {\(Z\)};
	\node[right of=A4] (C4) {\(Y\)};
	\node[below of=B4, yshift=-10mm] (R4) {\(R_Y\)};
	\draw[->] (C4) to (B4);	
	\draw[->] (A4) to (R4);
	\draw[->] (C4) to (R4);
	
	\node[above of=A1, xshift=-4mm, yshift=-5mm] () {1)};
	\node[above of=A2, xshift=-4mm, yshift=-5mm] () {2)};
	\node[above of=A3, xshift=-4mm, yshift=-5mm] () {3)};
	\node[above of=A4, xshift=-4mm, yshift=-5mm] () {4)};
	
	\end{tikzpicture}
	\end{center}
	\caption{Four missingness DAGs with identical missingness structures. In all DAGs, \(R_Y=R^{XYZ}\), as \(Y\) is the only variable containing missing values. DAGs 1 and 2: The true DAGs are such that oracle test-wise-deletion PC recovers the true CPDAG. DAGs 3 and 4: The conditional independence \(X\ind Y\mid Z\) is not identified under test-wise deletion, hence the CPDAG discovered by oracle test-wise-deletion PC using correct (in)dependence information will contain an edge between \(X\) and \(Y\).}
	\label{fig:identification}
\end{figure}
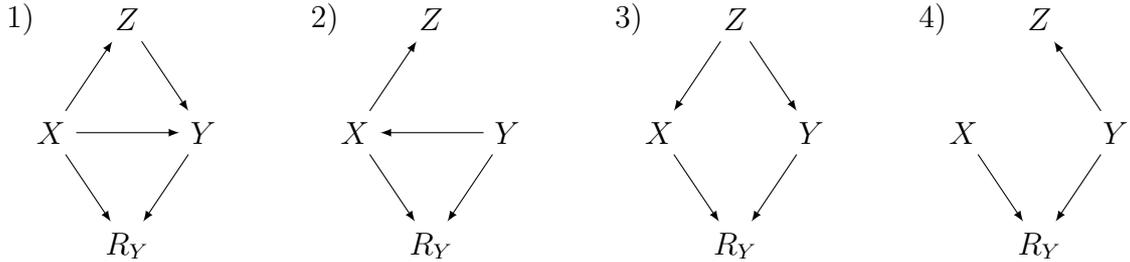

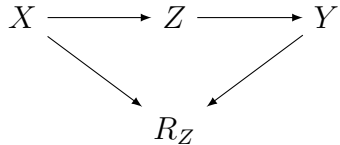
\begin{figure}[hb]
	\vspace{5mm}
	\begin{center}
		\begin{tikzpicture}[node distance=20mm, >=latex]
		\node (A) {\(X\)};
		\node[right of=A] (B) {\(Z\)};
		\node[right of=B] (C) {\(Y\)};
		\node[below of=B, yshift=5mm] (R) {\(R_Z\)};
		\draw[->] (A) to (B);	
		\draw[->] (B) to (C);
		\draw[->] (A) to (R);
		\draw[->] (C) to (R);
		
		\end{tikzpicture}
	\end{center}
	\caption{Example missingness DAG in which missingness depends on fully observed variables only, yet the conditional independence \(X\ind Y\mid Z\) is not identified under test-wise deletion. Oracle test-wise-deletion PC returns a fully connected graph.}
	\label{fig:MAR*}
\end{figure}

\newpage
\cite{Tuetal2019} (see also \citealp{Tuetal2020}) proposed two modifications of test-wise-deletion PC that can recover the correct CPDAG even if the admissible separator condition does not hold. Both aim at reconstructing relevant aspects of the full-data distribution. The first variant simulates values of all variables involved in the test based on models fitted to the observed data, the second variant re-weights the observed data. Both variants assume that no variable is a direct cause of its own missingness indicator, i.e.\ edges of the type \(V_i\rightarrow R_i\) are not allowed, and that there are no edges between the missingness indicators.

\subsubsection{Test-wise deletion vs.\ list-wise deletion}

Proposition~\ref{prop:twd} holds for list-wise deletion if \(R^{XY\mathbf{Z}}\) is replaced by \(R^{\mathbf{V}}=\mathbf{R}(\mathbf{V})\) in the admissible separator condition. The condition then requires for a pair \((X,Y)\) that \(R^{\mathbf{V}}\ind X\mid(Y,\mathbf{Z})\) or \(R^{\mathbf{V}}\ind Y\mid(X,\mathbf{Z})\), which is a stronger assumption than just requiring \(R^{XY\mathbf{Z}}\ind X\mid(Y,\mathbf{Z})\) or \(R^{XY\mathbf{Z}}\ind Y\mid(X,\mathbf{Z})\). It follows that under the assumptions of Proposition~\ref{prop:twd}, if oracle list-wise deletion PC recovers the true CPDAG, then oracle test-wise deletion does as well, but not the other way around. Consider now applying both variants to a given finite dataset. Then list-wise deletion PC uses only the completely observed data rows, while test-wise deletion also uses the incompletely observed rows for some of the conditional independence tests it performs. The graph discovered by test-wise-deletion PC is thus expected to be denser than the graph recovered by list-wise deletion PC, due to the larger power of some tests.

\subsubsection{Parametric assumptions}
So far, we have only considered non-parametric identification of conditional dependencies and independencies. In practice, parametric tests such as Fisher's \(z\)-test, which assumes that the variables follow a multivariate normal distribution, may be used. In that case, a complication arises for both list-wise and test-wise deletion, as the parametric assumptions need to hold conditionally on the response indicator being 1. As an example, suppose we have three variables \textit{income}, \textit{media} (measuring media consumption) and \textit{sedentary} (measuring sedentary behaviour), and assume that the missingness mechanism is such that people with a high media consumption are less likely to answer the media question, implying \textit{media}\(\rightarrow R^{media}\), while the other two variables are completely observed. When applying test-wise-deletion PC using Fisher's \(z\)-test to the incomplete data, we make the following assumptions: For testing \textit{income}\(\ind\)\textit{sedentary}, we assume that the full-data distribution of (\textit{income}, \textit{sedentary}) is normal, while for testing \textit{income}\(\ind\)\textit{media}, we assume that the conditional distribution of (\textit{income}, \textit{media}) given \(R^{media}=1\) is normal. Under the assumed missingness mechanism, these assumptions are incompatible: If the full-data distribution of \textit{media} is normal, then the observed distribution has a flattened right tail, since we assume that higher values are more likely to be missing. Hence at least one of these assumptions must be wrong, which potentially invalidates the type I error rate under the null hypothesis or decreases the power under the alternative.

\subsection{Multiple imputation}
\label{sec:MI}

Multiple imputation is a popular method for handling missing data especially in the context of regression analysis. It involves generating \(m\) predictions for each missing value using one of two strategies: For \textit{joint model imputation}, a joint distribution over all variables of interest is specified. Alternatively, separate models are specified for each incompletely observed variable given all other variables. This is called \textit{fully conditional specification} or \textit{multiple imputation by chained equations} (\textit{MICE}) and is more flexible than joint model imputation when it comes to different measurement scales. In either case, Bayesian regression models are fitted and predictions are drawn from the posterior predictive distribution(s) of the missing data given the observed data. The resulting \(m\) datasets are separately analysed using the same method that would have been used in the absence of missing values. Finally, the \(m\) results are pooled according to Rubin's rules \citep{Rubin1987} or other rules depending on the parameter of interest. 

Standard implementations of multiple imputation rely on the MAR assumption, although known MNAR mechanisms can be accommodated as well. In addition, it is required that the modelling assumptions made in the imputation phase do not contradict the assumptions made during the analysis. This is further discussed below.

Two approaches are conceivable for combining constraint-based causal discovery with multiple imputation. One would be to estimate and pool \(m\) graphs. However, it is not clear what a good pooling method would be. The other one is to pool at the test level, as proposed by \cite{Foraitaetal2020}: First, \(m\) imputed datasets are generated using standard multiple imputation techniques. Then causal discovery is applied with the following modification: For each test, the test statistic is calculated using each of the \(m\) datasets in turn, and the \(m\) test statistics are pooled using appropriate rules. The test decision is based on the pooled statistic before going to the next test. This way, a single estimated graph is obtained.

Rubin's rules are valid for pooling Wald-type test statistics such as the \(z\)-statistic of Fisher's \(z\)-test \citep{Rubin1987}. For likelihood ratio statistics such as those of the \(G^2\)-test and the CG-test, appropriate rules have been proposed by \cite{MengRubin1992}. See Appendix~\ref{app:condindtests} for details on both sets of rules. Thus, for these and similar tests no new methodology is required for the pooling step. The rules guarantee that under the null hypothesis of conditional independence, the rejection rate is below the nominal \(\alpha\) level. However, this assumes that an appropriate imputation model has been used. As discussed next, choosing the imputation models is more problematic in the context of causal discovery than in the regression context.

\subsubsection{Choosing the imputation model}
\label{sec:modelselection}

Rubin's rules \citep{Rubin1987}, as well as the rules by \cite{MengRubin1992}, were derived within the joint model framework and assuming that the imputation model and the analysis model are \textit{compatible}, meaning the models do not contradict each other \citep{Meng1994, Bartlettetal2015}. When using MICE, where imputation is based on a set of separate imputation models, a common joint distribution underlying all these models can exist only in special cases, e.g.\ when all imputation models are linear regression models or saturated logistic regression models \citep{Hughesetal2014}. In all other cases, the theoretical guarantees of the pooling rules do not apply, even though MICE has been found to be robust in many settings even in the absence of an underlying joint model (see \citealp{Hughesetal2014}, and the references therein).

In the context of causal discovery, two complications arise. One is that each conditional independence test assumes its own analysis model, and the different analysis models may contradict each other. This is not the case for causal discovery using Fisher's \(z\)-test only: here we can impute using either a multivariate normal joint model, or MICE with linear main effects regression in order to ensure compatibility \citep{Hughesetal2014}. We do not recommend using predictive mean matching \citep{MorrisWhiteRoyston2014}, as we found this method to lead to an increased type I error rate in several scenarios (results not shown). Similarly, for causal discovery using the \(G^2\)-test only, we can either use a multinomial joint model, or MICE with saturated (i.e.\ including all possible interactions) logistic regression \citep{Hughesetal2014}. In contrast, consider using the CG-test. The CG-distribution is not collapsible, i.e.\ if we assume a CG-distribution for all variables jointly, this does not imply that a given subset of the variables also follows a CG-distribution \citep{LauritzenWermuth1989, Lauritzen1990}. The different analysis models thus contradict each other in general, and a compatible imputation model does not exist. Similarly, the analysis models will often contradict each other if a mix of different tests is used.

The second complication is that the number of variables in causal discovery analyses is often large, but imputation processes becomes instable when too many variables or model terms are involved \citep{vanBuuren2018, HardtHerkeLeonhart2012}. Consider MICE using saturated logistic models: With \(10\) variables, the number of terms in each imputation model is \(2^{10}=1\,024\); with \(100\) variables, it equals \(2^{100}>10^{30}\). Some amount of model selection is necessary, but it is not clear what a good approach would be. For joint model imputation based on the multivariate normal distribution, it has been suggested to apply a ridge penalty (e.g.\ \citealp{Schafer1997, CarpenterKenward2013}), but this has not been generalised to other variable types. For discrete variables in particular, one could consider restricting the order of the interaction terms. Another idea is to use flexible imputation models, e.g.\ based on random forests \citep{DoovevanBuurenDusseldorp2014, Shahetal2014}.

\subsubsection{Hybrid procedure}
\label{sec:2step}

As an alternative to the above strategies for selecting the imputation models, we propose the following hybrid procedure. First, a preliminary graph skeleton is estimated using test-wise deletion, with a nominal \(\alpha\) larger than the one to be used in the actual analysis. \cite{Tuetal2019} showed that under the assumptions of Proposition~\ref{prop:twd}, the estimated graph will be a supergraph of the true skeleton (see Lemma~\ref{lemma:Tu}). In a second step, MICE is performed such that the imputation model for variable \(V\) contains only \(V\)'s neighbours and the neighbours of the neighbours. The rationale is that ideally, the imputation model for variable \(V\) would include all variables in the \textit{Markov blanket} of the node \(V\), which is defined as the set of \(V\)'s parents, children and `spouses', i.e.\ nodes with which \(V\) shares a common child. As the edges in the estimated preliminary skeleton are undirected, every neighbour of \(V\) is a potential child and every neighbour of a neighbour a potential `spouse'.

\subsection{Auxiliary information and noise -- when multiple imputation is expected to outperform test-wise deletion}

Both test-wise deletion and testing under multiple imputation yield type I error rates respecting the nominal significance level, under their respective assumptions. In addition, testing under multiple imputation has the potential to detect (conditional) associations with a higher power than test-wise deletion, as (i) no observations are deleted, and (ii) multiple imputation exploits information in the observed values about the missing values. However, there are also situations in which multiple imputation is not expected to outperform test-wise deletion, e.g.\ when the incomplete variable is in the conditioning set of the conditional independence test, as illustrated in Scenario B of the following simulation experiment. Moreover, when the number of variables in the imputation model(s) is large, the imputation process could be dominated by noise and become unstable, as illustrated in Scenario D.

\begin{example}
	\label{ex:twdMI}
	Consider the following causal graph and covariance matrix, implying \(X\not\ind Y\mid Z\):
	\begin{center}
	\begin{minipage}{0.3\textwidth}
		\begin{tikzpicture}[node distance=20mm, >=latex]
		\node (X) {\(Z\)};
		\node[above of=X] (Z) {\(X\)};
		\node[right of=Z] (Y) {\(Y\)};
		\node[right of=X] (A) {\(A\)};
		\node[right of=Y, xshift=-7mm] (N1) {\(N_1\)};
		\node[below of=N1, yshift=10mm] () {\(\vdots\)};
		\node[right of=A, xshift=-7mm] (N99) {\(N_{99}\)};
		\draw[->] (Z) to (Y);
		\draw[->] (Z) to (X);
		\draw[->] (X) to (Y);
		\draw[->] (X) to (A);
		\draw[->] (Y) to (A);
		\draw[->] (Z) to (A);
		\end{tikzpicture}
	\end{minipage}
	\begin{minipage}{0.49\textwidth}
	\[\bm{\Sigma}=
	\begin{pmatrix}
	1 & 0.2 & 0.2 & 0.5 & 0 & \dots & 0 \\
	0.2 & 1 & 0.2 & 0.2 & 0 & \dots & 0  \\
	0.2 & 0.2 & 1 & 0.2 & 0 & \dots & 0  \\
	0.5 & 0.2 & 0.2 & 1 & 0 & \dots & 0  \\
	0 & 0 & 0 & 0 & 1 & \dots & 0\\
	\vdots & \vdots & \vdots & \vdots & \vdots & \ddots & \vdots \\
	0 & 0 & 0 & 0 & 0 & 0 & 1\\
	\end{pmatrix}\]
	\end{minipage}
	\end{center}

	We generated \(n=50\) or \(n=500\) observations of \((X,Y,Z,A,N_1,\dots,N_{99})\sim\mathcal{N}(\mathbf{0},\bm{\Sigma})\) and then deleted and imputed values as follows. Scenario A: 10, 30, 50 or 70\,\% of the values of \(Z\) were made MCAR. Imputation was based on a linear regression of \(Z\) on \((X,Y)\). Scenarios B,C,D: 10, 30, 50 or 70\,\% of the values of \(X\) were made MCAR. Scenario B: Imputation was based on a linear model of \(X\) on \((Y,Z)\). Scenario C: Imputation was based on a linear model of \(X\) on \((Y,Z,A)\). Scenario D: Imputation was based on a linear model of \(X\) on \((Y,Z,A,N_1,\dots, N_{99})\). The number of imputations was \(m=100\). In all scenarios, \(X\ind Y\mid Z\) was tested (i) using Fisher's \(z\)-test (\(\alpha=0.05\)) with test-wise deletion and (ii) using Fisher's \(z\)-test (\(\alpha=0.05\)) on the multiply imputed data. Figure~\ref{fig:Ex1} shows the rejection rate (power) over 10\,000 replications.
\end{example}

\begin{figure}[h]
	\begin{center}
		\includegraphics{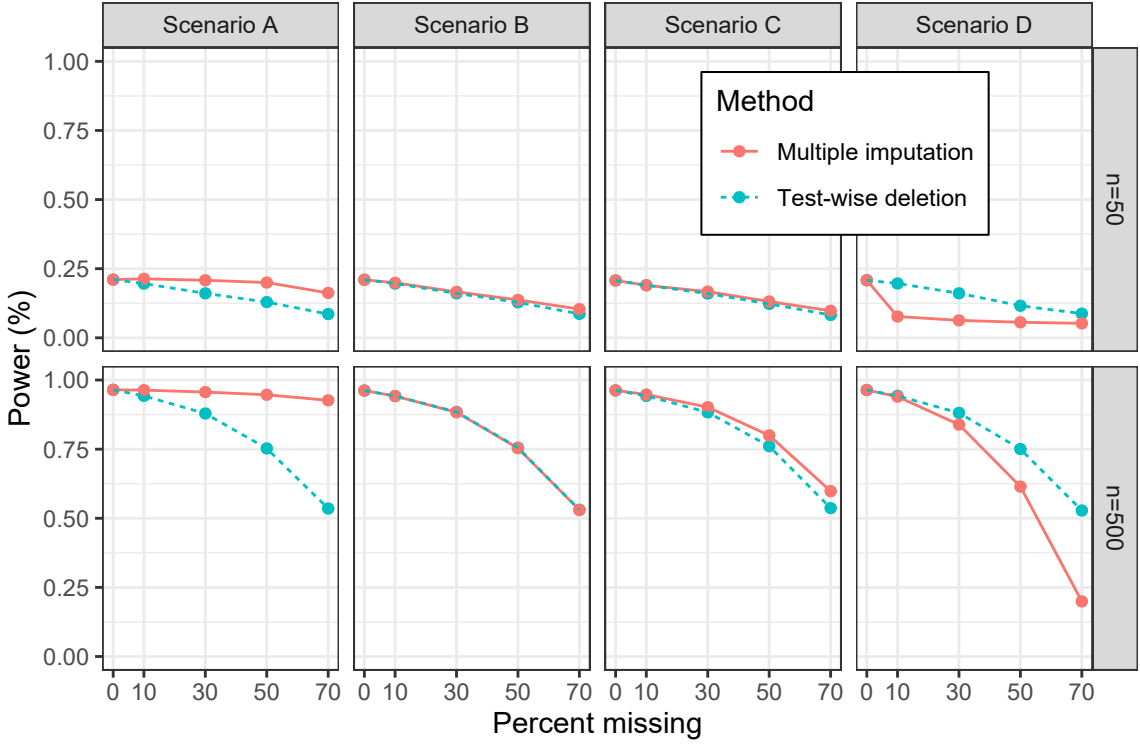}
		\caption{Power of Fisher's \(z\)-test combined with test-wise deletion versus multiple imputation. The null hypothesis is \(X\ind Y\mid Z\). \textbf{Scenario A}: Values in \(Z\) are MCAR. \textbf{Scenario B}: Values in \(X\) are MCAR. \textbf{Scenario C}: Values in \(X\) are MCAR, the imputation model includes an auxiliary variable. \textbf{Scenario D}: Values in \(X\) are MCAR, the imputation model includes an auxiliary variable and 99 noise variables.}
		\label{fig:Ex1}
	\end{center}
\end{figure}

Figure~\ref{fig:Ex1} shows that in Scenario A, multiple imputation successfully exploited information in \(X\) and \(Y\) to partially recover the missing information about \(Z\), resulting in a higher power for detecting \(X\not\ind Y\mid Z\). In contrast, multiple imputation did not result in a higher power in Scenario B, where missing values occurred in \(X\). This is a phenomenon well-known in the context of regression analysis: When the analysis model is a model for \(\mathrm{E}(X\mid Y,Z)\) and missingness occurs in \(X\), then multiple imputation using the imputation model \(\mathrm{E}(X\mid Y,Z)\) only adds noise to the analysis, hence restricting the analysis to the complete cases is preferred (\citealp{LittleRubin2002}, page 237; \citealp{Hughesetal2019}). Testing \(X\ind Y\mid Z\) using Fisher's \(z\)-test is conceptually equivalent to modelling \(\mathrm{E}(X\mid Y,Z)\) and testing for the coefficient of \(Y\) being zero. A different situation occurs when the imputation model includes additional variables not in the analysis model, such as the variable \(A\) in the simulation. In the context of regression analysis, such variables are called \textit{auxiliary} to the variables in the analysis model. In Scenario C, where the imputation model for \(X\) included \(Y\), \(Z\) and \(A\), the multiple imputation procedure successfully exploited information in \(A\), hence the power was (slightly) higher. In Scenario D, however, where the imputation model for \(X\) additionally included 99 noise variables, the noise outweighed the auxiliary information, hence the power was even lower than under test-wise deletion.

The above has consequences for causal discovery, where each variable can have different roles (variable of interest, conditioning variable, auxiliary variable, noise variable) relative to the different tests that are performed. First, multiple imputation is expected to benefit from graphs with strong associations between the variables, as the observed values then contain more information about the missing ones. Second, multiple imputation is expected to benefit from dense graphs. This is because during the PC-algorithm, the conditional independence test \(X\ind Y\mid\mathbf{Z}\) is only performed if \(X\) or \(Y\) still have more than \(|\mathbf{Z}|\) neighbours. Hence, fewer conditional tests will be performed if the graph is sparse, but as argued above, multiple imputation is especially effective when missing values occur in the conditioning variable(s). These two trends become visible in Illustration~\ref{ex:tiles} below. Third, as discussed in the previous section, variable selection on the imputation models is needed when the number of variables or terms in the imputation models is large, as otherwise the models are dominated by noise.

\begin{example}
	\label{ex:tiles}
	Random DAGs with 8 nodes each were generated using the \texttt{randomDAG} function from the \texttt{R} package \texttt{pcalg} \citep{R-pcalg}. The edge density parameter (probability of connecting a newly added node to a node already in the graph) was set to 0.1 (`very sparse'), 0.25 (`sparse'), 0.4 (`medium'), 0.55 (`dense') or 0.7 (`very dense'). The graphs were parameterised as linear structural models, where the edge weights \(w\) were chosen such that the power of Fisher's \(z\)-test (\(\alpha=0.05\)) for detecting the marginal dependence \(X\not\ind Y\) in the model \(X\stackrel{w}{\rightarrow}Y\) was 10\,\% (very weak), 30\,\% (weak), 50\,\% (medium), 70\,\% (strong) or 90\,\% (very strong). 500 observations were generated, and 10\,\% of the values were randomly deleted. Graphs were estimated by the PC-algorithm using Fisher's \(z\)-test (\(\alpha=0.05\)) (i) with test-wise deletion, (ii) on multiply imputed data (based on linear regressions, 100 imputations) and (iii) on the full data (for comparison). Figure~\ref{fig:Ex2} shows the results averaged over 1\,000 repetitions per scenario. The colour intensity is proportional to the difference in the relative Hamming distance, i.e.\ \(H_{\text{MI}}/H_{\text{full}} - H_{\text{twd}}/H_{\text{full}}\), where \(H_{\text{MI}}\), \(H_{\text{twd}}\) and \(H_{\text{full}}\) are the average Hamming distances (to the true graph) obtained using multiple imputation, test-wise deletion and the full data, respectively. The plot shows that under our simple data-generating model, multiple imputation yields better graph estimates in all but the extreme cases, and the advantage over test-wise deletion tends to be greatest in dense graphs with strong dependencies between variables.
\end{example}

\begin{figure}[]
	\begin{center}
		\includegraphics[width=12cm]{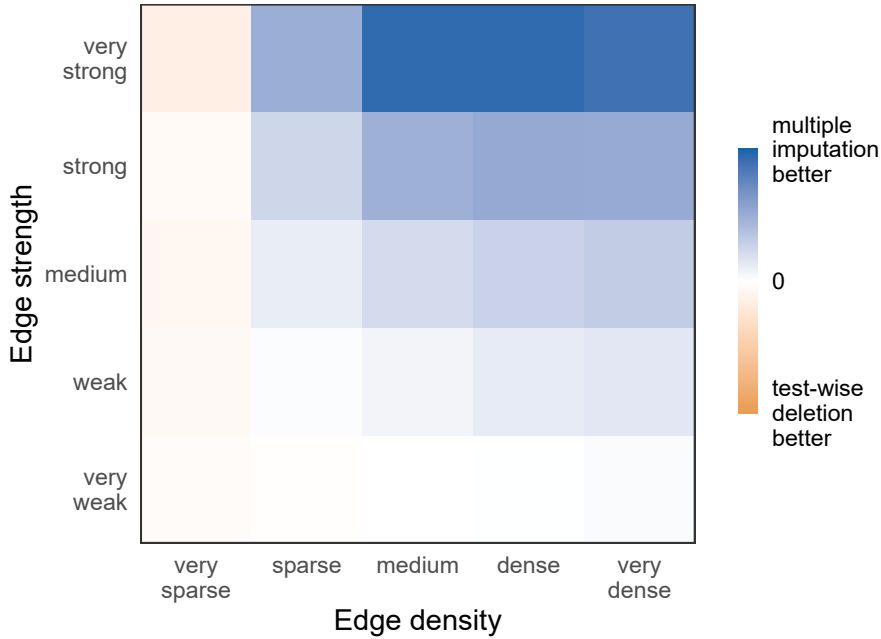}
		\caption{Relative performance of multiple imputation vs.\ test-wise deletion for discovering random graphs with 8 nodes and different edge densities and edge strengths. 10\,\% of data points were MCAR. The colour intensity is proportional to the relative Hamming distance (compared to using the full data); white means that the same relative Hamming distance was obtained for multiple imputation and test-wise deletion.}
		\label{fig:Ex2}
	\end{center}
\end{figure}

\FloatBarrier

\section{Detailed comparison on synthetic data}
\label{sec:simulation}

The aim of the simulation study was to compare the performance of test-wise deletion and multiple imputation with different imputation models, in order to help guide the choice between the different methods in practice. \texttt{R} code for replication can be found on GitHub (link will be provided after publication).

We considered 7 data-generating mechanisms (\textit{ECOLI}, \textit{MAGIC}, \textit{ASIA}, \textit{SACHS}, \textit{HEALTHCARE}, \textit{MEHRA}, \textit{ECOLI\_large}), 3 sample sizes (n=100, 1\,000, 5\,000) and 3 missingness mechanisms (`MCAR', `MAR', `MNAR'), yielding a total of 63 simulation scenarios. 

\subsection{Synthetic incomplete datasets}

Data were generated from benchmark causal graphs and their data-generating mechanisms according to the Bayesian Network Repository (\url{www.bnlearn.com/bnrepository}). Table~\ref{tab:dgm} summarises their key features. The \textit{ECOLI} graph is a subgraph of the \textit{ECOLI\_large} graph.

\begin{table}
	\caption{Data-generating mechanisms used in the simulation study. The footnotes indicate the variables of the selected subgraphs.}
	\label{tab:dgm}
	\begin{tabular}{lrrrrrl}
		\toprule
		& \multicolumn{3}{c}{\#Variables} & \#Edges & \#Categories & \\
		\cline{2-4}
		& total & Gaussian & discrete & & & \\
		\midrule
		\textit{ECOLI}\textsuperscript{1} & 12 & 12 & -- & 17 & & \\
		\textit{MAGIC}\textsuperscript{2} & 7 &7 & -- & 7 & & \\
		\textit{ASIA} & 8 & -- & 8 & 8 & 2 each & \\
		\textit{SACHS} & 11 & -- & 11 & 17 & 3 each & \\
		\textit{HEALTHCARE} & 7 & 4 & 3 & 9 & 2/3/3 & \\
		\textit{MEHRA}\textsuperscript{3} & 8 & 4 & 4 & 14 & 31/6/20/9 & \\
		\textit{ECOLI\_large} & 46 & 46 & -- & 70 & & \\
		\bottomrule
		\multicolumn{7}{l}{\textsuperscript{1} \footnotesize{b1191, cchB, eutG, fixC, ibpB, sucA, tnaA, yceP, yfaD, ygbD, ygcE, yjbO}}\\
		\multicolumn{7}{l}{\textsuperscript{2} \footnotesize{MIL, G1217, G257, G2208, G1338, G524, G1945 of \textit{MAGIC-NIAB}}}\\
		\multicolumn{7}{l}{\textsuperscript{3} \footnotesize{Zone, Type, Year, Region, co, pm10, pm2.5, so2}}\\
	\end{tabular}
\end{table}

In the \textit{ECOLI}, \textit{MAGIC}, \textit{ASIA}, \textit{SACHS}, \textit{HEALTHCARE} and \textit{MEHRA} scenarios, missing values were generated as follows. For `MCAR', 18\,\% of the values in the dataset were randomly chosen and deleted. Both multiple imputation and test-wise deletion are valid under this missingness mechanism. For `MAR', one or two groups of three or four variables each were chosen at random. Using the \texttt{ampute} function from the \texttt{mice} package \citep{R-mice}, missing values were generated such that exactly one variable per group was missing in each data row, and the probability of missingness depended on the values of the other two or three variables in the group. Values in one other randomly chosen variable were randomly deleted until an overall missingness proportion of 18\,\% was reached. Under this `MAR' mechanism, multiple imputation is valid, while test-wise deletion is not, as the admissible separator condition is not necessarily fulfilled for all pairs of variables. For `MNAR', we chose one fixed `key' variable and four to nine `subordinate' variables per graph. In the data rows with the \(q\)\,\% largest values of the `key' variable, the values of the `key' variable and all `subordinate' variables were deleted, where \(q\) was chosen such that the overall missingness proportion was 18\,\%. Under this `MNAR' mechanism, the admissible separator condition is satisfied for all  pairs of variables, hence test-wise deletion is valid, while multiple imputation is not.

In the \textit{ECOLI\_large} scenarios, missing values were generated in the same variables and using the same mechanisms as in the \textit{ECOLI} scenarios, leading to an overall missingness proportion of 4.7\,\% (instead of 18\,\%) in each scenario.

\subsection{Missing data methods}

The PC-stable algorithm as implemented in \texttt{pcalg} \citep{R-pcalg, ColomboMaathuis2014} was applied, using the following methods for dealing with the missing values: 1) List-wise deletion, i.e.\ data rows with missing observations were deleted before applying PC-stable. 2) Test-wise deletion using \texttt{gaussCItwd}, \texttt{disCItwd} or \texttt{mixCItwd} from the \texttt{micd} package. 3-4) Test-wise deletion with the (3) density or (4) permutation correction method by \cite{Tuetal2019} as implemented in the \texttt{MVPC} repository (\url{www.github.com/TURuibo/MVPC}; only available for continuous or binary data). 5-9) Conditional independence testing under multiple imputation using \texttt{gaussMItest}, \texttt{disMItest} or \texttt{mixMItest} from the \texttt{micd} package, where the imputations were generated using the \texttt{mice} package \citep{R-mice} with different imputation models, as follows: (5) each variable was imputed based on the variables in its Markov blanket (i.e.\ its parents, children and nodes with which it shares a common child) using linear or logistic regression imputation including all interaction terms (`oracle' multiple imputation; this is not possible to do in practice as the graph is not known, but is included here as a reference); (6) linear or logistic regression imputation including all interaction terms; (7) main effects linear or logistic regression imputation; (8) random forest imputation using the \texttt{rf} option \citep{DoovevanBuurenDusseldorp2014}; (9) random forest imputation using the \texttt{rfcont} or \texttt{rfcat} option from \texttt{CALIBERrfimpute} \citep{Shahetal2014}. 10) Missing values were singly imputed with the column mean (continuous data) or mode (discrete data) before applying PC-stable. For multiple imputation, we choose \(m=10\) imputations. Although this number is smaller than what is recommended in the literature \citep{vanBuuren2018, CarpenterKenward2013}, we found in preliminary simulations (not shown) that the test rejection rates do not change considerably when more imputations are added. We still recommend using \(m=100\) or higher in real applications. For methods (5) and (6), the highest order of interaction was set to 2 or 3 if required to reduce the runtime. For random forest imputation, we set the number of trees to 100, as we found this to improve the quality of the estimated graphs, compared to the default of 10 trees, in preliminary simulations (not shown).

In the \textit{ECOLI\_large} scenarios, we additionally included three versions (A, B, C) of the hybrid procedure proposed in Section~\ref{sec:2step}, where in step 1, the preliminary graph skeleton was estimated using \texttt{alpha=0.2}. In versions B and C, the skeleton search was stopped after all marginal independence tests had been performed, as the higher-order tests are expected to be less reliable. Additionally, in version C, the neighbours of the neighbours were ignored, in order to obtain even sparser imputation models.

\subsection{Evaluation criteria}
The performance was evaluated using the following metrics: number of edges in the estimated graph; proportion of discovered edges among the edges in the true CPDAG, ignoring edge orientation (recall); proportion of correctly discovered edges among the discovered edges, ignoring edge orientation (precision); number of edge insertion or deletions in order to transform the estimated graph into the true CPDAG, ignoring edge orientation (Hamming distance); and number of edge insertions, deletions or reversals in order to transform the estimated graph into the true CPDAG (structural Hamming distance; \citealp{TsamardinosBrownAliferis2006}).

\subsection{Results}
\label{sec:results}

The runtime was about three weeks on a 240-node high-performing computer cluster. Figure~\ref{fig:sim} provides a first overview of the results. It compares the performance of test-wise deletion without correction vs.\ multiple imputation based on linear models (Gaussian variables), logistic models including interaction terms (discrete variables) or the CALIBER random forest method (mixed variables). The horizontal position of the points in the figure is determined by the difference in the relative Hamming distance as defined in Illustration~\ref{ex:tiles}.

\begin{figure}[h]
	\includegraphics[width=\textwidth]{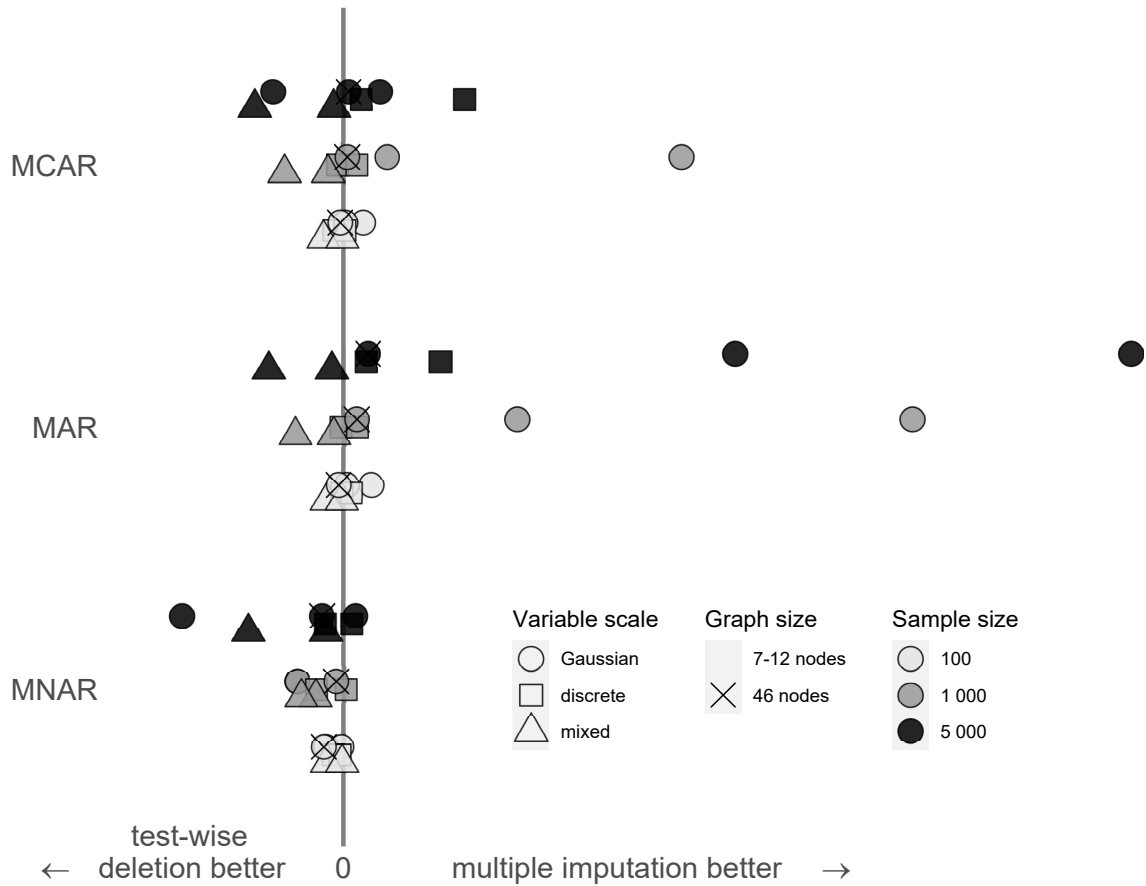}
	\caption{Overview over the simulation results. Shown on the x-axis is the difference in the relative Hamming distance (compared to using the full data).}
	\label{fig:sim}
\end{figure}

The following trends are apparent: For \(n=100\), there were virtually no differences in the performance of test-wise deletion vs.\ multiple imputation, and the differences were most pronounced for \(n=5\,000\). For Gaussian variables, multiple imputation outperformed test-wise deletion in almost all `MCAR' and `MAR' scenarios, but not in the `MNAR' scenarios. For discrete variables, the same trend can be observed, but the advantage of multiple imputation only occurred for \(n=5\,000\). For mixed variables, test-wise deletion outperformed multiple imputation in all simulation scenarios.

\begin{figure}
	\begin{center}
	\includegraphics[height=4.8cm, trim=0.3cm 0 0.45cm 0, clip]{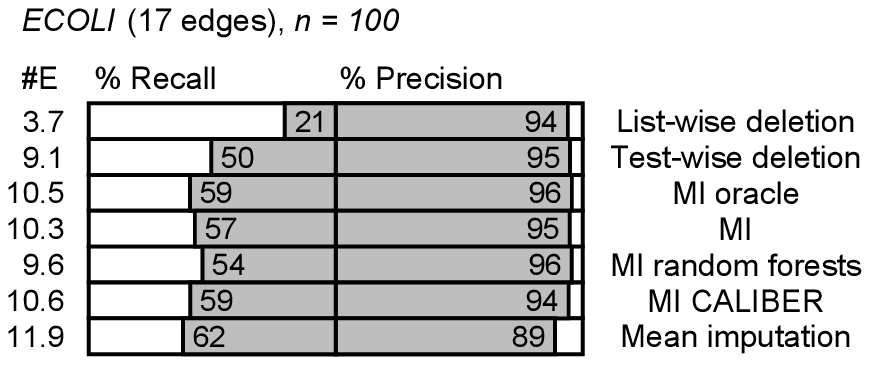}
	\includegraphics[height=4.8cm, trim=0.75cm 0 3.07cm 0, clip]{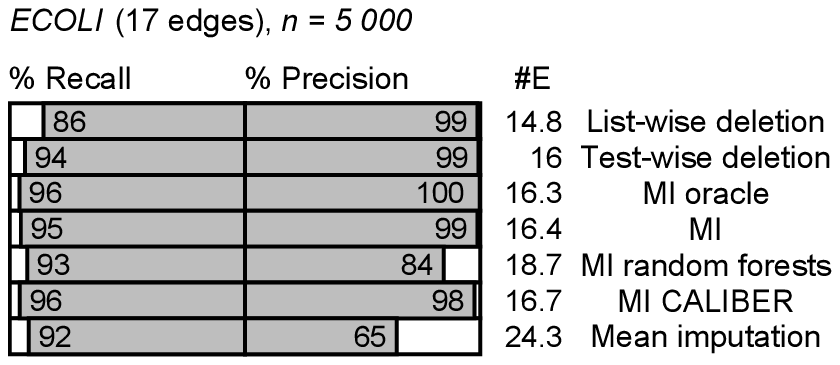}
	\includegraphics[height=4.8cm, trim=0.3cm 0 0.45cm 0, clip]{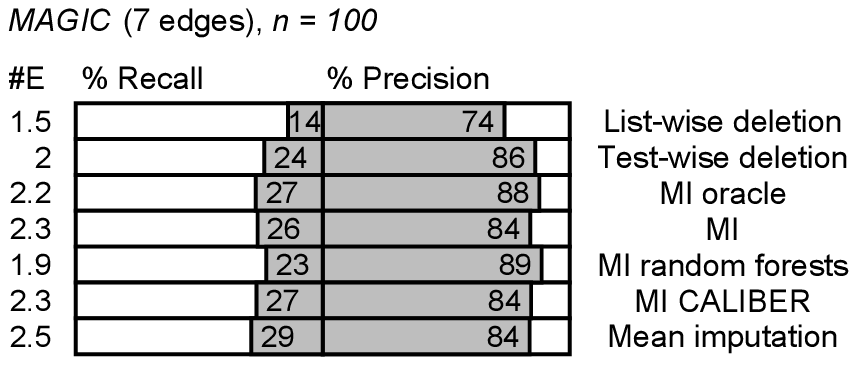}
	\includegraphics[height=4.8cm, trim=0.75cm 0 3.07cm 0, clip]{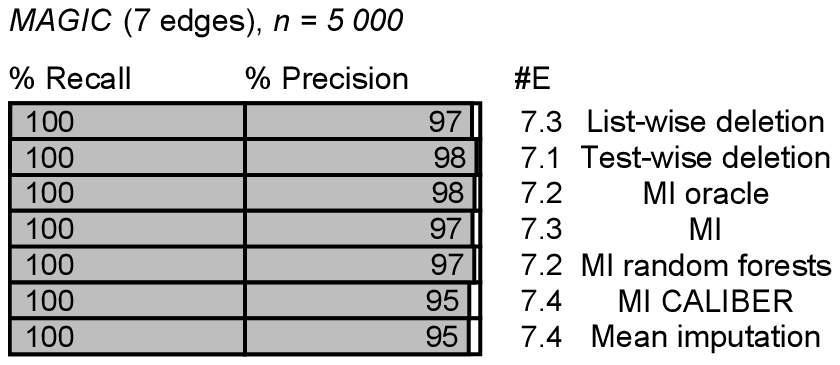}
	\includegraphics[height=5.1cm, trim=0.3cm 0 0.45cm 0, clip]{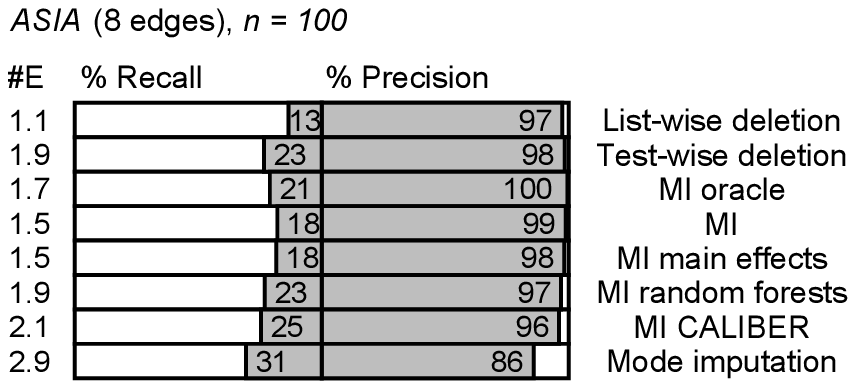}
	\includegraphics[height=5.1cm, trim=0.75cm 0 3.07cm 0, clip]{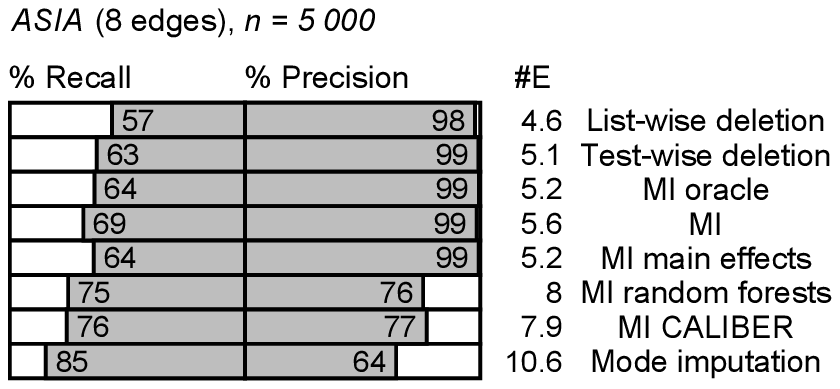}
	\includegraphics[height=5.1cm, trim=0.3cm 0 0.45cm 0, clip]{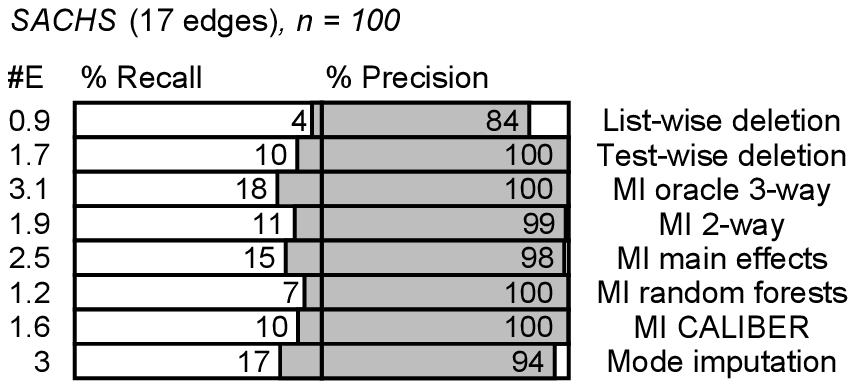}
	\includegraphics[height=5.1cm, trim=0.75cm 0 3.07cm 0, clip]{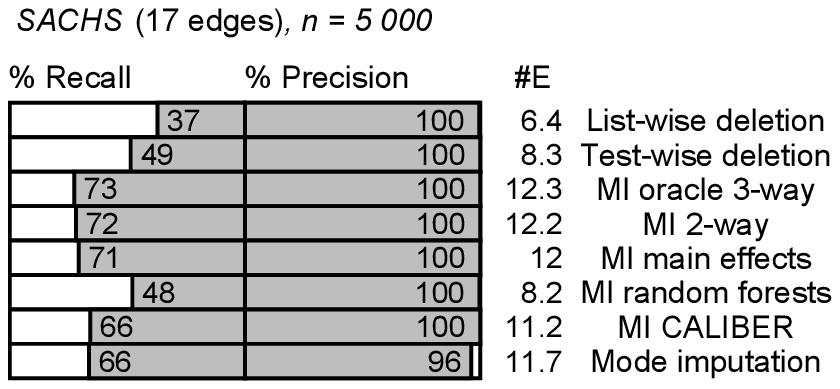}
	\end{center}
\caption{Simulation results, part I (\textit{ECOLI}, \textit{MAGIC}, \textit{ASIA}, \textit{SACHS}). Shown are the average edge recall (\% Recall), the average edge precision (\% Precision) and the average number of edges (\#E) in 1\,000 graphs estimated using the PC-algorithm combined with different methods for handling missing values. The sample size was either \(n=100\) (left) or \(n=5\,000\) (right) and missing values were generated using the `MCAR' mechanism described in the text. MI=multiple imputation.}
\label{fig:resI}
\end{figure}

\begin{figure}
\begin{center}
	\includegraphics[height=4.8cm, trim=0.3cm 0 0.45cm 0, clip]{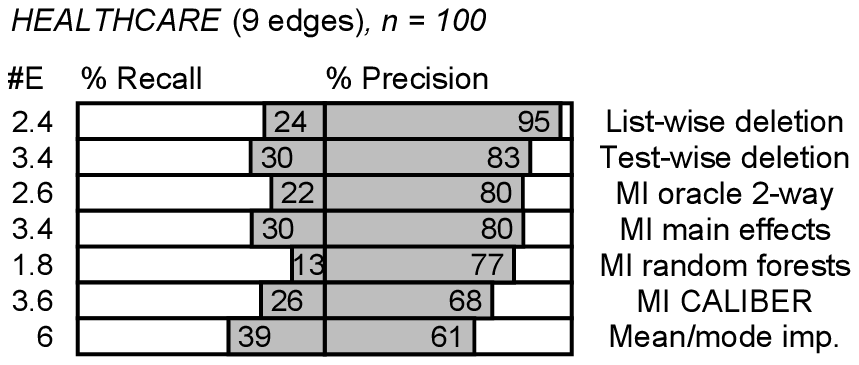}
	\includegraphics[height=4.8cm, trim=0.75cm 0 3.07cm 0, clip]{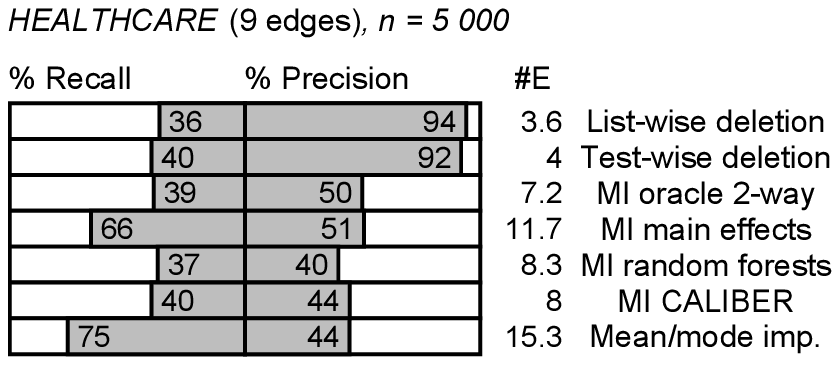}
	\includegraphics[height=4.4cm, trim=0.3cm 0 0.45cm 0, clip]{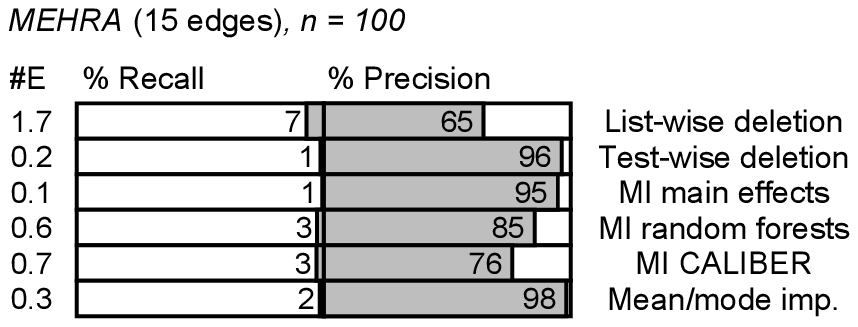}
	\includegraphics[height=4.4cm, trim=0.75cm 0 3.07cm 0, clip]{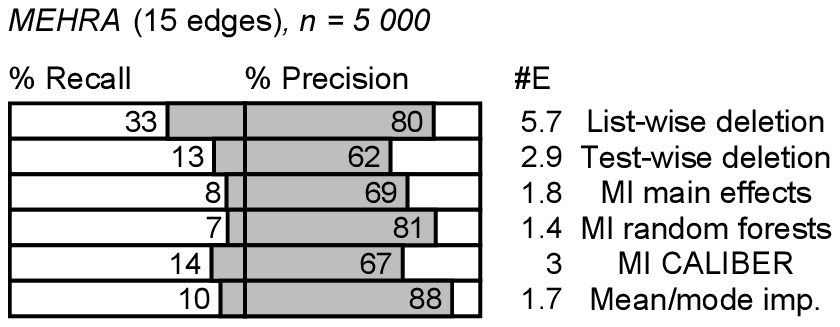}
	\includegraphics[height=5.9cm, trim=0.3cm 0 0.45cm 0, clip]{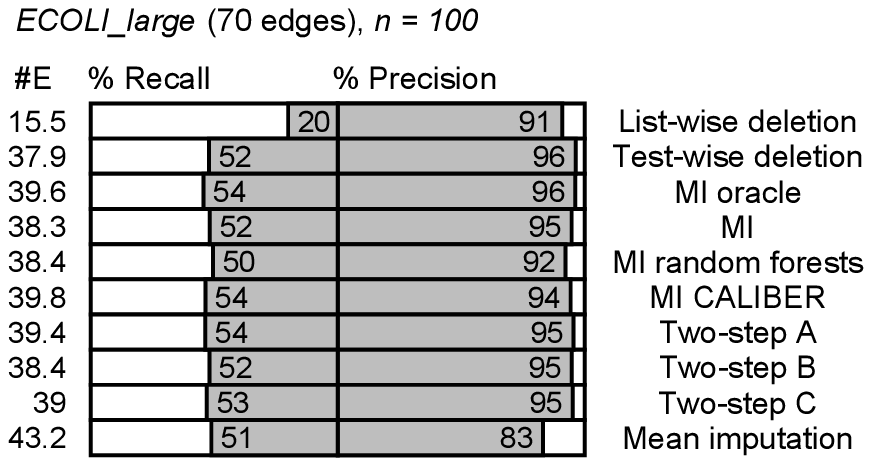}
	\includegraphics[height=5.9cm, trim=0.75cm 0 3.07cm 0, clip]{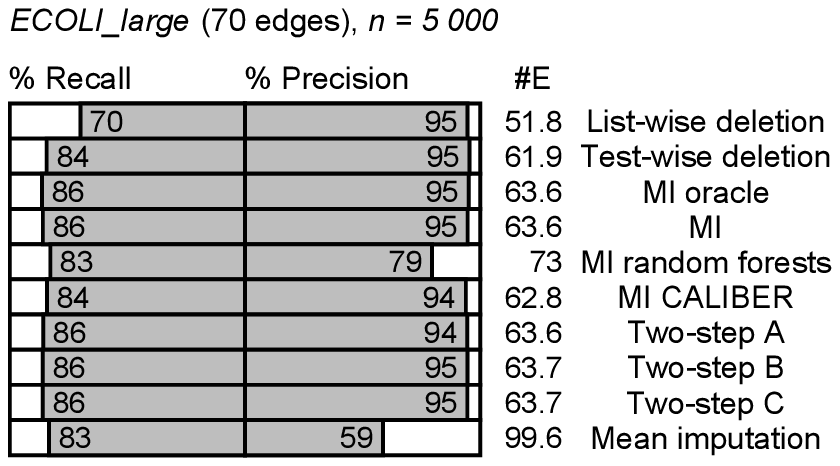}
\end{center}
\caption{Simulation results, part II (\textit{HEALTHCARE}, \textit{MEHRA}, \textit{ECOLI\_large}). Shown are the average edge recall (\% Recall), the average edge precision (\% Precision) and the average number of edges (\#E) in 1\,000 graphs estimated using the PC-algorithm combined with different methods for handling missing values. The sample size was either \(n=100\) (left) or \(n=5\,000\) (right) and missing values were generated using the `MCAR' mechanism described in the text. MI=multiple imputation.}
\label{fig:resII}
\end{figure}

More detailed results are shown in Figures \ref{fig:resI} and \ref{fig:resII}, and in the tables in the Online Supplement. The main observations are as follows: First, the edge recall was generally largest in the scenarios with Gaussian variables, smaller for discrete variables and very small for mixed variables. Further, while an average precision of more than 95\,\% was attained by a subset of the missing data methods in all Gaussian and discrete scenarios, this was not the case for the mixed scenarios. This is in line with earlier results using data without missing values \citep{AndrewsRamseyCooper2018} and indicates that causal discovery using mixed data is a particularly challenging task.

As expected, list-wise deletion resulted in sparse graphs with large (structural) Hamming distances, due to the low power. An exception occurred in the \textit{MEHRA} `MCAR' scenarios with \(n=100\) and \(m=5\,000\), where list-wise deletion resulted in denser graphs than test-wise deletion. This seemingly paradoxical behaviour can be explained as follows. The PC-algorithm starts with marginal tests and proceeds to conditional testing only if the nodes still have enough neighbours to be included in the conditioning set. Under list-wise deletion, only few edges remain after the marginal phase, hence the number of conditional tests performed is small. Under test-wise deletion, more edges survive the marginal phase, hence more conditional tests are performed, but this leads to the deletion of most remaining edges due to the very low power of the CG-test conditioning on the categorical \textit{MEHRA} variables with 6--31 categories.

Single imputation by the column mean or mode led to graphs that were `too large' (many edges but low precision). In order to understand why this happened, consider the structure \(X\rightarrow Y\rightarrow Z\), implying \(X\ind Z\mid Y\) but \(X\nind Z\). If \(Y\) contains missing values and these are replaced by the column mean or mode, it is very likely that after conditioning on the imputed version of \(Y\), there remains a residual association between \(X\) and \(Z\). However, the sample size is as large as if the data had been complete to begin with, and the fact that values were imputed is not taken into account by the testing procedure. Hence, the null hypothesis is rejected with a probability larger than the nominal test level, so that the resulting graph is more likely to contain an edge between \(X\) and \(Z\).

Test-wise deletion performed well overall. The results using the correction methods proposed by \cite{Tuetal2019} are not shown in Figures \ref{fig:resI} and \ref{fig:resII}, as they were very similar to those using test-wise deletion without correction. This is not surprising, as the missingness mechanisms chosen for the simulation did not require these corrections. \cite{Tuetal2019} and \cite{Tuetal2020} showed that if they are required (which is usually not known in real data analyses), using the corrections improves the performance; we conclude that if they are not required, the performance is at least not substantially worsened.

Concerning multiple imputation, we observed that parametric imputation using interaction terms was computationally infeasible (producing errors) in many repetitions, especially for the datasets with mixed variables. See the Online Supplement for more information. We obtained inconclusive results for the usefulness of the two variants of random forest imputation. In the scenarios with only Gaussian or only categorical variables, the \texttt{rf} variant often produced graphs with a lower precision and worse Hamming distance than parametric imputation, and the CALIBER variant often ranged between parametric and \texttt{rf} imputation in terms of different evaluation metrics. In the scenarios with mixed variables, the performance of the two random forest methods was usually similar and also comparable to that of parametric imputation. The main difference between the two random forest options lies in how they guarantee that the multiply imputed values are sufficiently different from each other (in order to properly account for the uncertainty in the missing values). Using the \texttt{rf} option, a specified number of trees is fitted and one tree is chosen at random. A prediction is made using this tree and the imputed value is randomly drawn from among the observed values that are in the same leaf as the prediction. The \texttt{CALIBERrfimpute} functions fit the random forest model on a bootstrap sample of the observed data. The imputed value is then either the best prediction plus a normal error (continuous case) or the predicted value from just one of the trees (discrete case). Based on our simulation results, we conclude that the CALIBER version is more appropriate for conditional independence testing, and we conjecture that the difference between the two versions is less pronounced when the goal is e.g.\ estimation of a regression coefficient.

The \textit{ECOLI\_large} results in Figure~\ref{fig:resII} demonstrate the potential of the hybrid method. For \(n=100\), test-wise deletion outperformed parametric multiple imputation (`MI' in the figure) in terms of the average precision and Hamming distance (34.9 for test-wise deletion vs.\ 35.8 for multiple imputation; see Online Supplement). Using the hybrid method A, the recall was as good as when using oracle multiple imputation, and the average Hamming distance was only 34.2. CALIBER random forest imputation also performed well and yielded an average Hamming distance of 34.8. We expect the differences to be larger in scenarios with even more variables.

\FloatBarrier

\section{Data application}
\label{sec:realdata}

To investigate the causal structure underlying the IDEFICS data, we used the \texttt{tpc} function from the \texttt{tpc} package (\url{www.github.com/bips-hb/tpc}), which is based on \texttt{pcalg} \citep{R-pcalg}, but offers additional options for integrating background knowledge (see \citealp{Andrewsetal2021}). We specified the following partial node ordering: (\textit{income}, \textit{isced}, \textit{bage}, \textit{migrant}, \textit{sex}) \(<\) \textit{smoke} \(<\) \textit{week} \(<\) \textit{bweight} \(<\) (\textit{formula}, \textit{hdiet}, \textit{bf}) \(<\) (\textit{age}, \textit{school}) \(<\) (\textit{bmi\_m}, \textit{fmeal}, \textit{yhei}) \(<\) (\textit{media}, \textit{mvpa}, \textit{sed}, \textit{sleep}, \textit{bmi}, \textit{homa}, \textit{wb}). In addition, we specified that \textit{sex} and \textit{age} are exogenous, i.e.\ do not have parent nodes. After obtaining rather sparse graphs in a test run, we set \texttt{alpha=0.1}. Missing data were dealt with using the following methods (in parentheses: name of the conditional independence test function used): list-wise deletion (\texttt{mixCItest}); test-wise deletion (\texttt{mixCItwd}); parametric multiple imputation based on main effects linear or logistic regression (\texttt{mixMItest}); random forest multiple imputation using \texttt{rf} in \texttt{mice} (\texttt{mixMItest}); random forest multiple imputation using \texttt{rf\_cont} or \texttt{rf\_cat} from the \texttt{CALIBERrfimpute} package (\texttt{mixMItest}); single imputation by the column mean or mode (\texttt{mixCItest}). For multiple imputation, we used 100 imputations and 100 trees where applicable. In order to get an impression of the variability of the estimated graphs, the whole analysis was repeated 50 times on bootstrap samples of the original data \citep{Pigeotetal2015}.

\begin{figure}
	\begin{tikzpicture}[>=latex]	
	\draw [draw, fill=lightgray!35] (-0.4,-0.2) rectangle (15.5,21.6);
	
	\node [rotate=90, align=center] at (0,18.5) {List-wise};
	\node [rotate=90, align=center] at (0.4,18.5) {deletion};
	\node [rotate=90, align=center] at (0,15.2) {Test-wise};
	\node [rotate=90, align=center] at (0.4,15.2) {deletion};
	\node [rotate=90, align=center] at (-0.1,8.4) {Multiple imputation};
	\draw (0.2,3.4) to (0.2,13.4);
	\node [rotate=90, align=center] at (0.5,11.8) {main effects};
	\node [rotate=90, align=center] at (0.5,8.4) {random forests};
	\node [rotate=90, align=center] at (0.5,5.0) {CALIBER};
	\node [rotate=90, align=center] at (0,1.6) {Single};
	\node [rotate=90, align=center] at (0.4,1.6) {imputation};
	\node [] at (1.7, 22.2-1) {(I)};
	\node [] at (1.7, 21.8-1) {Parent};
	\node [] at (1.7, 21.54-1) {demo-};
	\node [] at (1.7, 21.18-1) {graphics};
	\node [] at (4.1, 22.2-1) {(II)};
	\node [] at (4.1, 21.8-1) {Pregnancy};
	\node [] at (6.2, 22.2-1) {(III)};
	\node [] at (6.2, 21.8-1) {Early};
	\node [] at (6.2, 21.5-1) {life};
	\node [] at (8.2, 22.2-1) {(IV)};
	\node [] at (8.2, 21.8-1) {Child};
	\node [] at (8.2, 21.54-1) {demo-};
	\node [] at (8.2, 21.18-1) {graphics};
	\node [] at (10.2, 22.2-1) {(V)};
	\node [] at (10.2, 21.8-1) {Family};
	\node [] at (10.2, 21.54-1) {character-};
	\node [] at (10.2, 21.18-1) {istics};
	\node [] at (13.7, 22.2-1) {(VI)};
	\node [] at (13.7, 21.8-1) {Child};
	\node [] at (13.7, 21.5-1) {behaviour};
	\draw (12,21.2-1) to (15.2,21.2-1);
	\draw (12,21.2-1) to (12,21.1-1);
	\draw (15.2,21.1-1) to (15.2,21.2-1);
	\draw [draw, fill=white] (0.8,0) rectangle (15.3,3.2);
	\draw [draw, fill=white] (0.8,3.4) rectangle (15.3,6.6);
	\draw [draw, fill=white] (0.8,6.8) rectangle (15.3,10.0);
	\draw [draw, fill=white] (0.8,10.2) rectangle (15.3,13.4);
	\draw [draw, fill=white] (0.8,13.6) rectangle (15.3,16.8);
	\draw [draw, fill=white] (0.8,17.0) rectangle (15.3,20.0);
	\node [] at (1.7,2.8) (income) {income};
	\node [] at (1.7,1.9) (isced) {~~~~isced};
	\node [] at (1.7,1.0) (bage) {~bage};
	\node [] at (1.7,0.4) (migrant) {migrant};
	\node [] at (4.2,2.6) (smoke) {smoke};
	\node [] at (4.2,2.2) (week) {week};
	\node [] at (4.2,1.3) (bweight) {bweight};
	\node [] at (4.2,0.4) (sex) {sex};
	\node [] at (6.2,2.1) (formula) {~~~formula};
	\node [] at (6.0,1.6) (hdiet) {hdiet};
	\node [] at (6.3,0.9) (bf) {bf};
	\node [] at (8.3,1.6) (age) {age};
	\node [] at (8.3,0.7) (school) {school};
	\node [] at (10.2,2.7) (bmim) {bmi\_m};
	\node [] at (10.2,1.8) (yhei) {yhei};
	\node [] at (10.2,0.9) (fmeal) {f.meal};
	\node [] at (12.7,2.8) (media) {media};
	\node [] at (12.7,1.6) (mvpa) {mvpa};
	\node [] at (12.7,1.2) (sed) {sed};
	\node [] at (12.7,0.6) (sleep) {sleep};
	\node [] at (14.5,2.3) (bmi) {bmi};
	\node [] at (14.5,1.4) (homa) {homa};
	\node [] at (14.5,0.4) (wb) {wb};
	\draw [->, bend right=12] (school) to (sleep);
	\draw [->, out=-50, in=180] (age) to (sleep);
	\draw [->] (age) to (school);
	\draw [->] (age) to (media);
	\draw [->] (age) to (mvpa);
	\draw [->, out=-25, in=20] (formula) to (bf);
	\draw [->] (week) to (bweight);
	\draw [->] (sex) to (bweight);
	\draw [->, bend left=70] (sex) to (week);
	\draw [->, out=60, in=180, looseness=0.7] (bweight) to (bmim);
	\draw [->, out=60, in=-170, looseness=0.5] (bweight) to (bmi);
	\draw [-, out=45, in=-90, looseness=0] (isced) to (income);
	\draw [-, out=-45, in=90, looseness=0] (isced) to (bage);
	\draw [->, out=-35, in=-160, looseness=1.7] (isced) to (bf);
	\draw [-, transform canvas={xshift=-3mm}] (income) to (bage);
	\draw [-, transform canvas={xshift=-6mm}] (migrant) to (income);
	\draw [->] (bmim) to (bmi);
	\draw [-] (homa) to (bmi);
	\draw [->, out=5, in=172, looseness=0.4] (income) to (bmim);
	\draw [->, out=7, in=173, looseness=0.4] (income) to (media);
	\draw [-] (fmeal) to (yhei);
	\draw [->] (yhei) to (media);
	\draw [->, out=8, in=-172, looseness=1.5] (migrant) to (media);
	\draw [->, out=-10, in=-170, looseness=0.2] (migrant) to (wb);
	
	\node [] at (1.7,2.7+17) (income1) {income};
	\node [] at (1.7,1.9+17) (isced1) {~~~isced};
	\node [] at (1.7,1.1+17) (bage1) {~~~bage};
	\node [] at (1.7,0.3+17) (migrant1) {migrant};
	\node [] at (4.2,2.7+17) (smoke1) {smoke};
	\node [] at (4.2,2.2+17) (week1) {week};
	\node [] at (4.2,1.2+17) (bweight1) {bweight};
	\node [] at (4.2,0.3+17) (sex1) {sex};
	\node [] at (6.2,2.7+17) (formula1) {formula};
	\node [] at (6.2,2.1+17) (hdiet1) {hdiet};
	\node [] at (6.2,1.4+17) (bf1) {bf};
	\node [] at (8.2,2.7+17) (age1) {age};
	\node [] at (8.2,1.7+17) (school1) {school};
	\node [] at (10.2,2.7+17) (bmim1) {bmi\_m};
	\node [] at (10.2,2.1+17) (yhei1) {yhei};
	\node [] at (10.2,1.2+17) (fmeal1) {f.meal};
	\node [] at (12.7,2.7+17) (media1) {media};
	\node [] at (12.7,1.8+17) (mvpa1) {mvpa};
	\node [] at (12.7,1.0+17) (sed1) {sed};
	\node [] at (12.7,0.3+17) (sleep1) {sleep};
	\node [] at (14.5,2.2+17) (bmi1) {bmi};
	\node [] at (14.5,1.5+17) (homa1) {homa};
	\node [] at (14.5,0.6+17) (wb1) {wb};
	\draw [-, transform canvas={xshift=-5mm}] (migrant1) -- (income1);
	\draw [->, out=5, in=180, looseness=0.5] (migrant1) to (bmi1);
	\draw [->, out=-5, in=180] (bage1) to (wb1);
	\draw [->, out=-10, in=180] (bage1) to (sleep1);
	\draw [->] (week1) to (bweight1);
	\draw [->, out=0, in=200] (week1) to (formula1);
	\draw [->, out=-30, in=20, looseness=1] (formula1) to (bf1);
	\draw [->] (age1) to (school1);
	\draw [->] (yhei1) to (media1);
	\draw [-] (sed1) to (mvpa1);
	
	\node [] at (1.7,2.7+13.7) (income2) {income};
	\node [] at (1.7,1.9+13.7) (isced2) {~~~isced};
	\node [] at (1.7,1.0+13.7) (bage2) {~~~bage};
	\node [] at (1.7,0.3+13.7) (migrant2) {migrant};
	\node [] at (4.2,2.7+13.7) (smoke2) {smoke};
	\node [] at (4.2,2.2+13.7) (week2) {week};
	\node [] at (4.2,1.3+13.7) (bweight2) {bweight};
	\node [] at (4.2,0.4+13.7) (sex2) {sex};
	\node [] at (6.2,2.6+13.7) (formula2) {formula};
	\node [] at (5.9,2.0+13.7) (hdiet2) {hdiet};
	\node [] at (6.2,0.5+13.7) (bf2) {bf};
	\node [] at (8.2,1.4+13.7) (age2) {age};
	\node [] at (8.2,0.5+13.7) (school2) {school};
	\node [] at (9.8,2.6+13.7) (bmim2) {bmi\_m};
	\node [] at (10.3,1.8+13.7) (yhei2) {yhei};
	\node [] at (10.3,0.9+13.7) (fmeal2) {f.meal};
	\node [] at (12.9,2.7+13.7) (media2) {media};
	\node [] at (12.7,2.0+13.7) (mvpa2) {mvpa};
	\node [] at (12.7,1.1+13.7) (sed2) {sed};
	\node [] at (12.7,0.5+13.7) (sleep2) {sleep};
	\node [] at (14.5,2.3+13.7) (bmi2) {bmi};
	\node [] at (14.5,1.4+13.7) (homa2) {homa};
	\node [] at (14.5,0.6+13.7) (wb2) {wb};
	\draw [-, out=-30, in=20, looseness=1] (formula2) to (bf2);
	\draw [-] (income2) to (isced2);
	\draw [-, out=-20, in=20] (income2) to (bage2);
	\draw [-, transform canvas={xshift=-5mm}] (income2) to (migrant2);
	\draw [-, out=-120, in=120, looseness=1.2] (isced2) to (migrant2);
	\draw [-] (isced2) to (bage2);
	\draw [->, very thick, out=-50, in=-160, looseness=0.25] (bage2) to (wb2);
	\draw [->, out=10, in=170, looseness=0.2] (income2) to (media2);
	\draw [->, out=5, in=-175, looseness=0.5] (migrant2) to (media2);
	\draw [->, very thick, out=30, in=-165, looseness=0.5] (yhei2) to (media2);
	\draw [->, very thick, out=-40, in=170, looseness=1.1] (yhei2) to (wb2);
	\draw [->] (wb2) to (sleep2);
	\draw [->, very thick, out=20, in=-170, looseness=0.8] (age2) to (media2);
	\draw [->, very thick] (age2) to (school2);
	\draw [->] (school2) to (sleep2);
	\draw [->, very thick, out=-5, in=-160] (age2) to (mvpa2);
	\draw [->] (sed2) to (mvpa2);
	\draw [->, very thick, out=-10, in=-130, looseness=0.7] (sex2) to (mvpa2);
	\draw [-] (yhei2) to (fmeal2);
	\draw [->, out=-35, in=180, looseness=1] (week2) to (fmeal2);
	\draw [->, very thick] (week2) to (bweight2);
	\draw [->, very thick] (sex2) to (bweight2);
	\draw [->, very thick, out=10, in=-170] (bweight2) to (bmim2);
	\draw [->, very thick, out=10, in=180] (bweight2) to (bmi2);
	\draw [->, very thick, out=-5, in=175] (bmim2) to (bmi2);
	\draw [->] (bmi2) to (homa2);
	
	\node [] at (1.7,2.6+10.3) (income4) {income};
	\node [] at (1.7,1.9+10.3) (isced4) {~~~isced};
	\node [] at (1.8,1.1+10.3) (bage4) {~~~bage};
	\node [] at (1.7,0.3+10.3) (migrant4) {migrant};
	\node [] at (4.2,2.7+10.3) (smoke4) {smoke};
	\node [] at (4.2,2.2+10.3) (week4) {week};
	\node [] at (4.2,1.3+10.3) (bweight4) {bweight};
	\node [] at (4.2,0.4+10.3) (sex4) {sex};
	\node [] at (6.2,2.6+10.3) (formula4) {formula};
	\node [] at (5.9,2.0+10.3) (hdiet4) {hdiet};
	\node [] at (6.2,0.8+10.3) (bf4) {bf};
	\node [] at (8.2,1.6+10.3) (age4) {age};
	\node [] at (8.2,0.5+10.3) (school4) {school};
	\node [] at (10.2,2.7+10.3) (bmim4) {bmi\_m};
	\node [] at (10.2,1.7+10.3) (yhei4) {yhei};
	\node [] at (10.2,0.4+10.3) (fmeal4) {f.meal};
	\node [] at (12.5,2.4+10.3) (media4) {media};
	\node [] at (12.3,1.7+10.3) (mvpa4) {mvpa};
	\node [] at (13.9,1.0+10.3) (sed4) {sed};
	\node [] at (12.7,0.5+10.3) (sleep4) {sleep};
	\node [] at (14.5,2.3+10.3) (bmi4) {bmi};
	\node [] at (14.5,1.4+10.3) (homa4) {homa};
	\node [] at (14.5,0.6+10.3) (wb4) {wb};
	\draw [-, transform canvas={xshift=-5mm}] (migrant4) to (income4);
	\draw [-, out=120, in=-120, looseness=1.2] (migrant4) to (isced4);
	\draw [->, out=15, in=140, looseness=0.15] (income4) to (bmi4);
	\draw [->] (isced4) to (smoke4);
	\draw [->, out=-30, in=180] (isced4) to (bf4);
	\draw [->, out=-15, in=-170, looseness=0.9] (bage4) to (bf4);
	\draw [->, very thick, out=-50, in=-160, looseness=0.28] (bage4) to (wb4);
	\draw [->, out=40, in=-40] (bf4) to (formula4);
	\draw [->] (wb4) to (sleep4);
	\draw [->, very thick, out=-25, in=170] (yhei4) to (wb4);
	\draw [->, very thick] (yhei4) to (media4);
	\draw [->, very thick, out=20, in=180] (age4) to (media4);
	\draw [->, out=-30, in=170] (age4) to (sleep4);
	\draw [->, very thick] (age4) to (school4);
	\draw [->, very thick, out=15, in=170] (age4) to (mvpa4);
	\draw [->] (media4) to (homa4);
	\draw [->] (homa4) to (bmi4);
	\draw [->, very thick, out=5, in=160, looseness=0.7] (bmim4) to (bmi4);
	\draw [->, very thick, out=5, in=165, looseness=0.85] (bweight4) to (bmi4);
	\draw [->, very thick, out=10, in=-170] (bweight4) to (bmim4);
	\draw [->] (school4) to (sed4);
	\draw [<->] (mvpa4) to (sed4);
	\draw [->, very thick] (week4) to (bweight4);
	\draw [->, very thick] (sex4) to (bweight4);
	\draw [->, out=170, in=-160] (sex4) to (week4);
	\draw [->, very thick, out=0, in=-170] (sex4) to (mvpa4);
	
	\node [] at (1.7,2.7+6.9) (income5) {income};
	\node [] at (1.7,1.9+6.9) (isced5) {~~~isced};
	\node [] at (1.7,1.0+6.9) (bage5) {~~~bage};
	\node [] at (1.7,0.3+6.9) (migrant5) {migrant};
	\node [] at (4.2,2.7+6.9) (smoke5) {smoke};
	\node [] at (4.2,2.2+6.9) (week5) {week};
	\node [] at (4.2,1.3+6.9) (bweight5) {bweight};
	\node [] at (4.2,0.4+6.9) (sex5) {sex};
	\node [] at (6.2,2.7+6.9) (formula5) {formula};
	\node [] at (5.9,2.2+6.9) (hdiet5) {hdiet};
	\node [] at (6.2,0.9+6.9) (bf5) {bf};
	\node [] at (8.0,1.6+6.9) (age5) {age};
	\node [] at (8.0,0.8+6.9) (school5) {school};
	\node [] at (10.2,2.7+6.9) (bmim5) {bmi\_m};
	\node [] at (10.2,1.7+6.9) (yhei5) {yhei};
	\node [] at (10.2,0.4+6.9) (fmeal5) {f.meal};
	\node [] at (12.5,2.4+6.9) (media5) {media};
	\node [] at (12.7,1.8+6.9) (mvpa5) {mvpa};
	\node [] at (12.7,1.0+6.9) (sed5) {sed};
	\node [] at (12.7,0.5+6.9) (sleep5) {sleep};
	\node [] at (14.5,2.7+6.9) (bmi5) {bmi};
	\node [] at (14.5,1.4+6.9) (homa5) {homa};
	\node [] at (14.5,0.5+6.9) (wb5) {wb};
	\draw [-, transform canvas={xshift=-5mm}] (migrant5) to (income5);
	\draw [-, out=120, in=-120, looseness=1.2] (migrant5) to (isced5);
	\draw [-] (income5) to (isced5);
	\draw [->] (isced5) to (smoke5);
	\draw [-, out=-40, in=40] (formula5) to (bf5);
	\draw [->, very thick, out=-40, in=-160, looseness=0.28] (bage5) to (wb5);
	\draw [->, very thick, out=-35, in=170] (yhei5) to (wb5);
	\draw [->, very thick] (yhei5) to (media5);
	\draw [->, very thick, out=20, in=-175] (age5) to (media5);
	\draw [->, very thick] (age5) to (school5);
	\draw [->, out=-30, in=170] (age5) to (sleep5);
	\draw [->, very thick, out=15, in=170] (age5) to (mvpa5);
	\draw [-] (sed5) to (mvpa5);
	\draw [->, very thick, out=0, in=-160, looseness=1.1] (sex5) to (mvpa5);
	\draw [->, out=170, in=-160] (sex5) to (week5);
	\draw [->, very thick] (sex5) to (bweight5);
	\draw [->, very thick] (week5) to (bweight5);
	\draw [->, very thick, out=20, in=-170, looseness=0.6] (bweight5) to (bmim5);
	\draw [->, very thick, out=15, in=-178] (bweight5) to (bmi5);
	\draw [->, very thick] (bmim5) to (bmi5);
	\draw [->] (bmi5) to (homa5);
	
	\node [] at (1.7,2.6+3.5) (income6) {income};
	\node [] at (1.7,1.7+3.5) (isced6) {isced};
	\node [] at (1.7,1.0+3.5) (bage6) {bage};
	\node [] at (1.7,0.3+3.5) (migrant6) {migrant};
	\node [] at (4.2,2.7+3.5) (smoke6) {smoke};
	\node [] at (4.2,2.2+3.5) (week6) {week};
	\node [] at (4.2,1.3+3.5) (bweight6) {bweight};
	\node [] at (4.2,0.4+3.5) (sex6) {sex};
	\node [] at (6.2,2.7+3.5) (formula6) {formula};
	\node [] at (5.8,2.2+3.5) (hdiet6) {hdiet};
	\node [] at (6.2,1.0+3.5) (bf6) {bf};
	\node [] at (8.2,1.6+3.5) (age6) {age};
	\node [] at (8.2,0.5+3.5) (school6) {school};
	\node [] at (10.2,2.7+3.5) (bmim6) {bmi\_m};
	\node [] at (10.2,1.7+3.5) (yhei6) {yhei};
	\node [] at (10.2,0.4+3.5) (fmeal6) {f.meal};
	\node [] at (12.7,2.4+3.5) (media6) {media};
	\node [] at (12.7,1.8+3.5) (mvpa6) {mvpa};
	\node [] at (13.9,1.3+3.5) (sed6) {sed};
	\node [] at (12.6,0.7+3.5) (sleep6) {sleep};
	\node [] at (14.5,2.7+3.5) (bmi6) {bmi};
	\node [] at (14.5,1.8+3.5) (homa6) {homa};
	\node [] at (14.5,0.6+3.5) (wb6) {wb};
	\draw [-] (income6) to (isced6);
	\draw [->] (isced6) to (smoke6);
	\draw [->, very thick] (sex6) to (bweight6);
	\draw [->, out=170, in=-160] (sex6) to (week6);
	\draw [->, very thick, out=0, in=-170] (sex6) to (mvpa6);
	\draw [->, very thick] (week6) to (bweight6);
	\draw [->, out=170, in=-35, looseness=0] (sed6) to (mvpa6);
	\draw [->, very thick, out=15, in=170] (age6) to (mvpa6);
	\draw [->] (age6) to (sleep6);
	\draw [->, very thick] (age6) to (school6);
	\draw [->, very thick, out=20, in=-175] (age6) to (media6);
	\draw [->, very thick] (yhei6) to (media6);
	\draw [->, very thick] (yhei6) to (wb6);
	\draw [->, out=-5, in=-172] (fmeal6) to (wb6);
	\draw [->, very thick, out=-40, in=-160, looseness=0.28] (bage6) to (wb6);
	\draw [-, bend left=45] (fmeal6) to (bmim6);
	\draw [->, very thick, out=20, in=-170, looseness=0.6] (bweight6) to (bmim6);
	\draw [->, very thick, out=15, in=-178, looseness=0.8] (bweight6) to (bmi6);
	\draw [->, very thick] (bmim6) to (bmi6);
	\draw [->] (bmi6) to (homa6);
	\draw [-, out=-40, in=40] (formula6) to (bf6);
	\end{tikzpicture}
	\caption{Estimated IDEFICS graphs. Bi-directed edges indicate that the direction could not be determined due to conflicting information in the data. Bold edges are present in all graphs estimated under test-wise deletion or multiple imputation.}
	\label{fig:ideficsgraphs}
\end{figure}
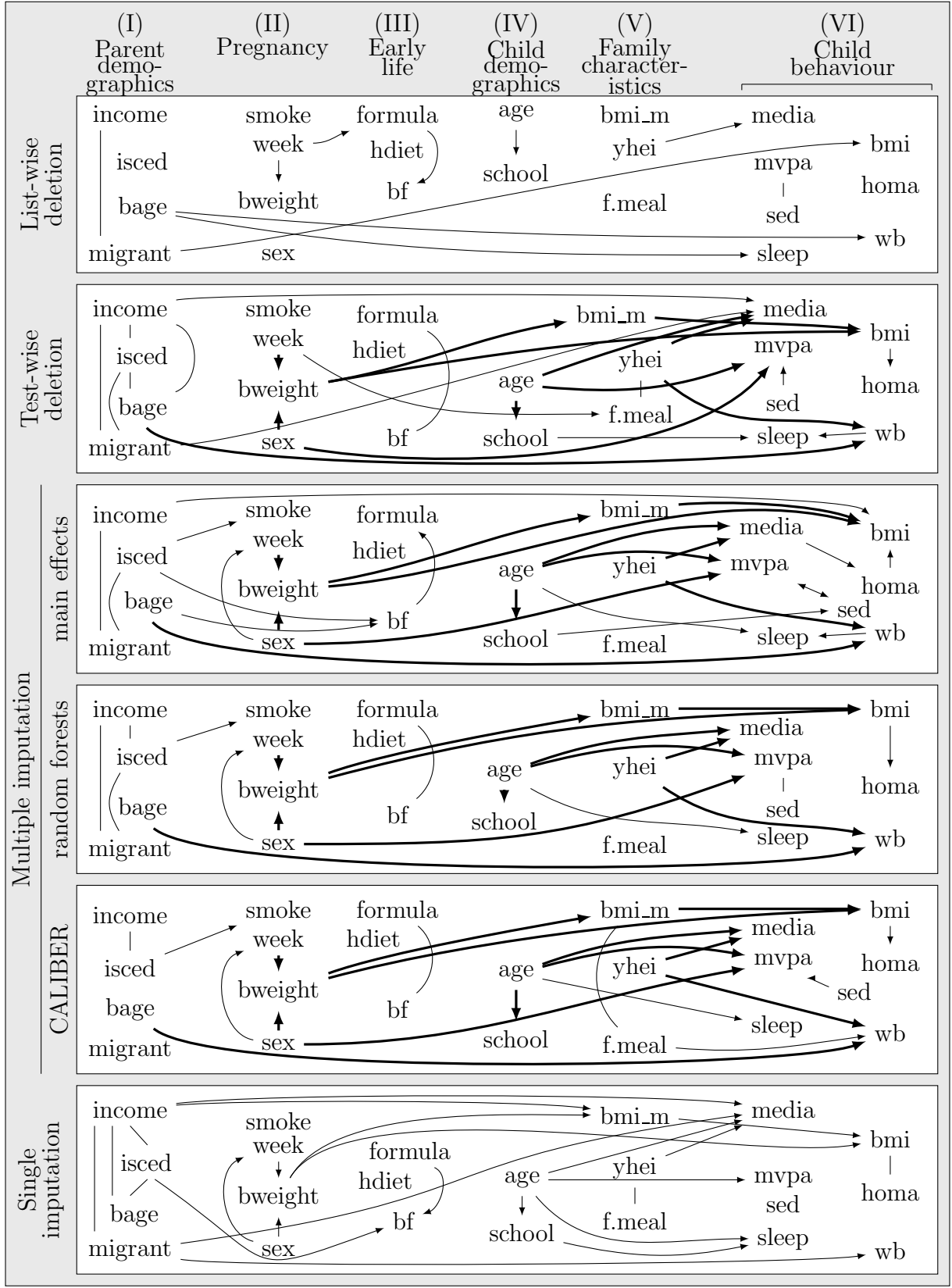

The graphs estimated in the main analysis are shown in Figure~\ref{fig:ideficsgraphs}. All discovered graphs were sparser than what might be expected based on expert knowledge \citep{VandenbroeckGoossensClemens2017}. Possible reasons could be the small sample size, violations of the faithfulness assumption, or deviations from the CG assumption. Consequently, the absence of edges should be interpreted with care.

Table~\ref{tab:nedges} compares the total number of edges and the number of edges adjoining nodes of `critical' variables with more than 20\,\% missing values, i.e.\ \textit{mvpa}, \textit{sed}, \textit{sleep} and \textit{homa}, in the main analysis and the bootstrap analyses. The numbers reveal, first of all, that the variability among the bootstrap samples was rather large, which again might be explained by the relatively small sample size. For random forest and single imputation, the number of edges obtained in the main analysis was smaller than the minimum number of edges obtained using the bootstrap samples. This is a known phenomenon and a correction has been proposed for score-based causal discovery \citep{SteckJaakkola2003}, but we are not aware of a correction method for the PC-algorithm. In line with the simulation results, list-wise deletion led to the sparsest graphs, while the densest graphs in the bootstrap analysis were discovered using single imputation. The multiple imputation methods tended to discover more edges adjoining `critical' nodes than test-wide deletion. This might be because the sample size available for tests containing the `critical' variables, where many values are missing, is rather small under test-wise deletion.

In the Online Supplement, we include diagnostic plots for the multiple imputation procedures in the main analysis. Based on visual inspection of the convergence plots, the algorithm converged in all three cases. The random forest (\texttt{rf}) method was most successful in generating imputed values with a distribution matching that of the observed values. Figure~\ref{fig:diagnostics} illustrates this for the \textit{wb} (well-being) variable. The distributions of the imputed values generated by the parametric and CALIBER methods are more symmetric. This may indicate that the (\texttt{rf}) method is better able to predict the missing values. However, as discussed previously \citep{Shahetal2014}, and as also witnessed in the simulation study, this does not necessarily mean that the graphs estimated using random forest (\texttt{rf}) imputation are closer to the truth.

\renewcommand{\arraystretch}{1.5} 

\begin{table}
\caption{Number of edges in the graphs estimated from the IDEFICS data (in parentheses: minimum, average and maximum number of edges in the bootstrap analyses). The critical edges are those adjoining nodes \textit{mvpa}, \textit{sed}, \textit{sleep} or \textit{homa}, for which large proportions of values were missing.}
\label{tab:nedges}
	\begin{tabular}{|l|rr|}
	\hline
	 & Total number of edges & Number of critical edges \vspace{-3mm}\\
	 & (min, average, max) & (min, average, max) \\
 	\hline
	List-wise deletion (lwd) & \cellcolor{lightgray!35} 10 (6, 15.9, 25) & \cellcolor{lightgray!35} 2 (1, 4.4, 8) \\
	Test-wise deletion (twd) & \cellcolor{lightgray!35} 26 (20, 27.3, 34) & \cellcolor{lightgray!35} 6 (2, 4.9, 8) \\
	Multiple imputation & \cellcolor{lightgray!35} & \cellcolor{lightgray!35} \\
	-- main effects (MI) & \cellcolor{lightgray!35} 26 (18, 25.3, 31) & \cellcolor{lightgray!35} 7 (3, 4.9, 7) \\
	-- random forests (rfMI) & \cellcolor{lightgray!35} 21 (22, 30.4, 37) & \cellcolor{lightgray!35} 5 (4, 8.2, 13) \\
	-- CALIBER (cMI) & \cellcolor{lightgray!35} 21 (20, 28.7, 37) & \cellcolor{lightgray!35} 5 (4, 7.5, 11) \\
	Single imputation (sing) & \cellcolor{lightgray!35} 24 (25, 32.7, 41) & \cellcolor{lightgray!35} 4 (4, 7.0, 13) \\
	\hline
	\end{tabular}
\end{table}

\begin{figure}[!h]	
	\includegraphics[height=6.2cm, trim= 0 0 8mm -10mm, clip]{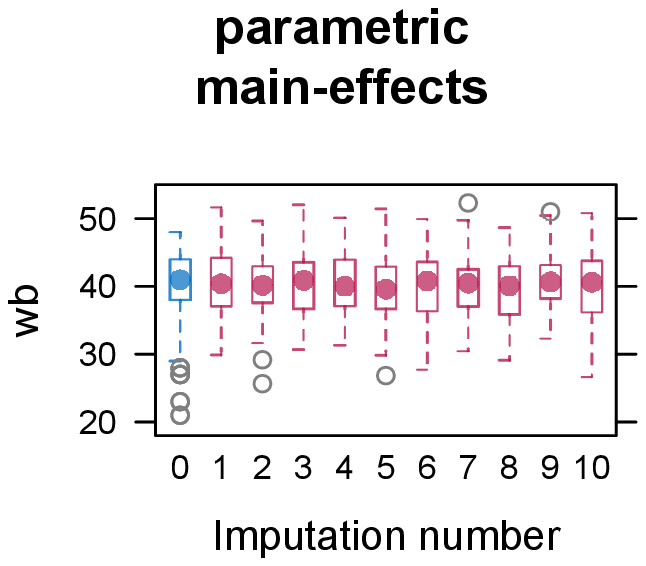}
	\includegraphics[height=6.2cm, trim= 10mm 0 8mm -10mm, clip]{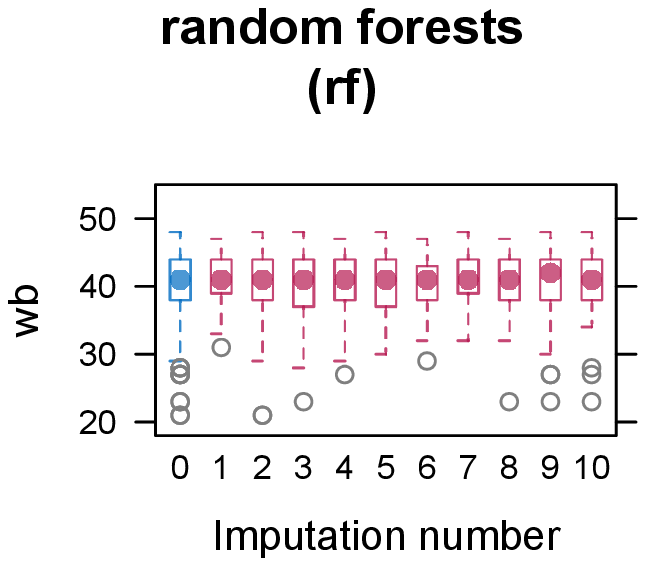}
	\includegraphics[height=6.2cm, trim= 10mm 0 8mm -10mm, clip]{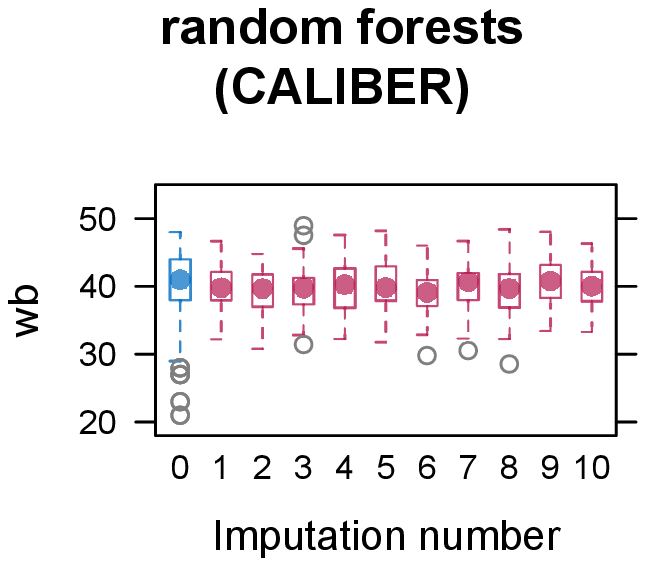}
	
	\includegraphics[height=5.5cm, trim= 0 0 8mm 5mm, clip]{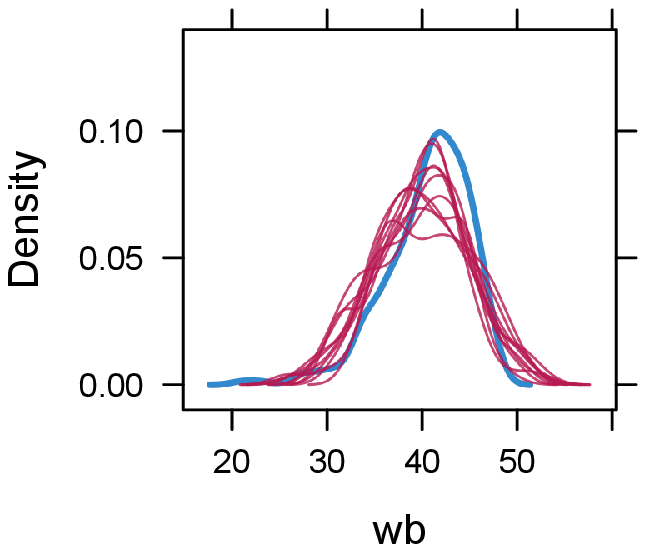}
	\includegraphics[height=5.5cm, trim= 10mm 0 8mm 5mm, clip]{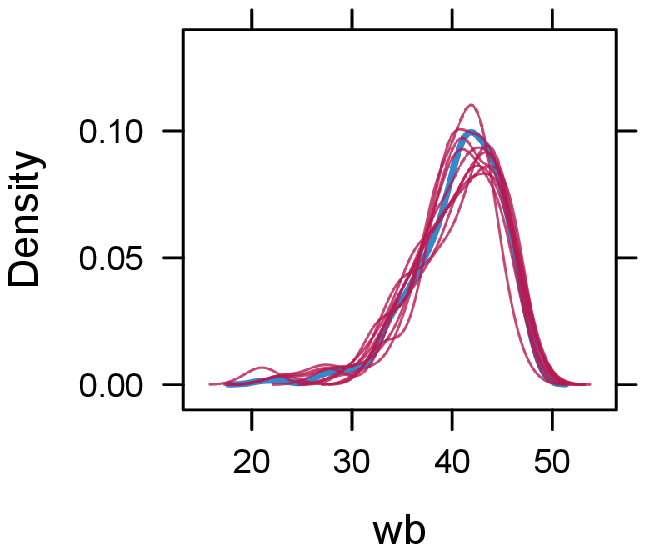}
	\includegraphics[height=5.5cm, trim= 10mm 0 8mm 5mm, clip]{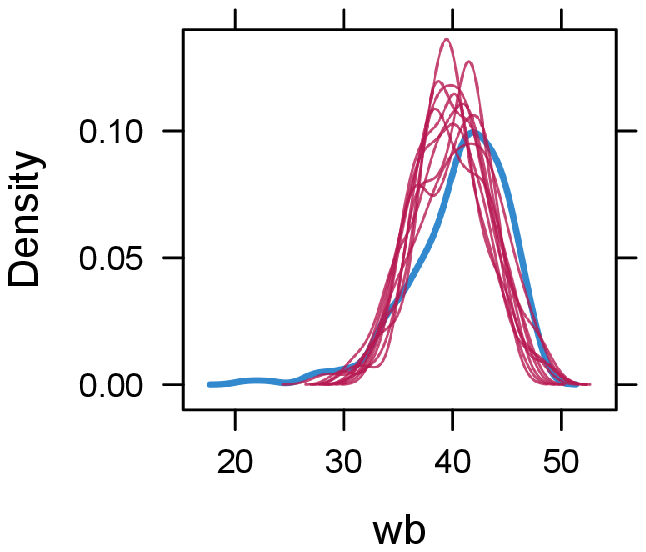}
	\caption{Diagnostic plots for variable \textit{wb} (well-being). Shown are the distribution of the observed values (blue boxplot and curves) and the distribution of the values generated by 10 randomly chosen imputations (red boxplots and curves), for the three different imputation methods.}
	\label{fig:diagnostics}
\end{figure}

\normalsize

\FloatBarrier

\section{Conclusions}
\label{sec:conclusions}

In this paper, we investigated test-wise deletion and multiple imputation for dealing with missing values in constraint-based causal discovery. Test-wise deletion relies on faithful observability and the admissible separator condition, whereas multiple imputation requires the missing values to be MAR. Both assumptions are implied by the stronger MCAR but are otherwise difficult to justify in practice.

In our empirical comparisons, we confirmed that test-wise deletion and multiple imputation are clearly superior to list-wise deletion and single imputation. We also demonstrated that while multiple imputation outperforms test-wise deletion in settings with small graphs and Gaussian variables, there is no overall best approach in realistically complex settings with a larger number of variables especially when these are a mix of continuous and discrete measurements. Random forest imputation or the hybrid method we proposed might be useful especially in settings with 50 or more variables, but comprehensive comparisons are difficult due to the long runtime of all three methods involved (causal discovery, multiple imputation and random forests).

An alternative missing value method for causal discovery not considered in this paper is inverse probability weighting \citep{GainShpitser2018}. Likelihood-based approaches such as Expectation Maximisation can be used with score-based causal discovery \citep{Friedman1997, Scutari2020} but are not straightforward to combine with constraint-based algorithms (see \citealp{Sokolovaetal2017}, for a first idea assuming a joint nonparanormal distribution).

Future research should address model selection of the imputation models in MICE. This is relevant also outside the area of causal discovery, but the literature on this topic is surprisingly scarce \citep{Noghrehchietal2021}. Finally, reliably learning (causal) graphs from data with mixed measurement scales remains a challenge especially with the additional complication of missing values.

\section*{Acknowledgements}
We gratefully acknowledge financial support by the German Research Foundation (DFG---Project DI 2372/1-1).

\FloatBarrier

\appendix
\section{Conditional independence testing}
\label{app:condindtests}

In this appendix, we provide details about the three conditional independence tests we focussed on in this work. We first review how each test is implemented when complete data are available, and then describe how they can be applied to multiple imputed data.

\subsection*{Fisher's \(z\)-test}

Consider a random vector \((X,Y,Z_1\dots,Z_s)^T\in\mathbb{R}^{s+2}\) with covariance matrix \(\bm{\Sigma}\). The \textit{partial correlation} between \(X\) and \(Y\) given \(\mathbf{Z}=(Z_1\dots,Z_s)\) is defined as
\[\rho_{XY.\mathbf{Z}}=\frac{p_{12}}{\sqrt{p_{11}}\sqrt{p_{22}}},\]
where \(p_{ij}\) is the \((i,j)\)-the element of the precision matrix \(\mathbf{P}=\bm{\Sigma}^{-1}\). The corresponding \textit{empirical partial correlation} can be estimated from \(n\) observations of \((X,Y,\mathbf{Z})^T\) as
\[\hat{\rho}_{XY.\mathbf{Z}}=\frac{\hat{\bm{\varepsilon}}_X^T\hat{\bm{\varepsilon}}_Y}{\sqrt{\hat{\bm{\varepsilon}}_X^T\hat{\bm{\varepsilon}}_X}\sqrt{\hat{\bm{\varepsilon}}_Y^T\hat{\bm{\varepsilon}}_Y}},\]
where \(\hat{\bm{\varepsilon}}_X\) is the vector of residuals after regressing \(X\) on \(\mathbf{Z}\), and \(\hat{\bm{\varepsilon}}_Y\) is the vector of residuals after regressing \(Y\) on \(\mathbf{Z}\).

For Fisher's \(z\)-test \citep{Fisher1924}, it is assumed that \((X,Y,\mathbf{Z})^T\) follows a multivariate normal distribution. Then \(\rho_{XY.\mathbf{Z}}=0\) if and only if \(X\ind Y\mid\mathbf{Z}\); this is the null hypothesis of Fisher's \(z\)-test. The test statistic is
\begin{equation}
\label{eq:z}
z(\hat{\rho}_{XY.\mathbf{Z}})=\frac{1}{2}\, \mathrm{ln}\left(\frac{1+\hat{\rho}_{XY.\mathbf{Z}}}{1-\hat{\rho}_{XY.\mathbf{Z}}}\right).
\end{equation}
Under the multivariate normal assumption, \(z(\hat{\rho}_{XY.\mathbf{Z}})\) is asymptotically normal with variance \(1/(n-s-3)\), and mean zero under the null hypothesis.


\subsection*{Fisher's \(z\)-test under multiple imputation}

Fisher's \(z\)-test can be applied to multiply imputed data using Rubin's rules, as follows (\citealp{Schafer1997}, page 109; \citealp{Foraitaetal2020}).

Consider \(M\) completed datasets obtained by multiple imputation, and let \(z^{(m)}(\hat{\rho}_{XY.\mathbf{Z}})\) be the \(z\)-statistic calculated according to Equation~(\ref{eq:z}) from the \(m\)-th imputed dataset, \(m=1,\dots,M\). The pooled test statistic is
\[\bar{z}(\hat{\rho}_{XY.\mathbf{Z}})=\frac{1}{M}\sum_{m=1}^M z^{(m)}(\hat{\rho}_{XY.\mathbf{Z}}).\]
The variance of \(\bar{z}(\hat{\rho}_{XY.\mathbf{Z}})\) is estimated as 
\begin{equation}
\label{eq:pooledvariance}
T_{XY.\mathbf{Z}}=\overline{W}_{XY.\mathbf{Z}}+\left(1+\frac{1}{M}\right)B_{XY.\mathbf{Z}},
\end{equation}
which has two components: \(\overline{W}_{XY.\mathbf{Z}}\) is the average \textit{within-imputation variance} and is calculated as
\[\overline{W}_{XY.\mathbf{Z}}=\frac{1}{M}\sum_{m=1}^M \frac{1}{n-s-3}=\frac{1}{n-s-3}.\]
The extra variance due to the missing values is captured in the \textit{between-imputation variance}
\[B_{XY.\mathbf{Z}}=\frac{1}{M-1}\sum_{m=1}^M \left[ z^{(m)}(\hat{\rho}_{XY.\mathbf{Z}})-\bar{z}(\hat{\rho}_{XY.\mathbf{Z}})\right]^2.\]
The term \(\left(1+\frac{1}{M}\right)\) in Equation~(\ref{eq:pooledvariance}) adjusts for the fact that only a finite number \(M\) of imputations was drawn.

Under the null hypothesis \(\rho_{XY.\mathbf{Z}}=0\), \(\bar{z}(\hat{\rho}_{XY.\mathbf{Z}})/\sqrt{T_{XY.\mathbf{Z}}}\) approximately follows a Student's \(t\)-distribution with degrees of freedom given by
\[\nu=(M-1)\left[1+\frac{\overline{W}_{XY.\mathbf{Z}}}{(1+M^{-1})B_{XY.\mathbf{Z}}}\right]^2.\]

\subsection*{The \(G^2\)-test}

Consider a vector \((X,Y,Z_1,\dots,Z_s)^T\) of categorical random variables, and define \(\mathbf{Z}=(Z_1,\dots,Z_s)\). The sets of values that \(X\), \(Y\) and \(\mathbf{Z}\) can take are denoted by \(\mathcal{X}\), \(\mathcal{Y}\) and \(\mathcal{Z}\), respectively. The vector \((X,Y,\mathbf{Z})^T\) thus defines a 3-way contingency table. Denote by \(\theta_{xy\mathbf{z}}\) the probability of observing \((x,y,\mathbf{z})^T\), for \(x\in\mathcal{X}\), \(y\in\mathcal{Y}\), \(\mathbf{z}\in\mathcal{Z}\). This corresponds to one cell in the contingency table. Further, denote the marginal probabilities with respect to \(X\) and \(Y\), respectively, as \(\theta_{+y\mathbf{z}}=\sum_{x\in\mathcal{X}}\theta_{xy\mathbf{z}}\) for \(y\in\mathcal{Y}\), \(\mathbf{z}\in\mathcal{Z}\), and \(\theta_{x+\mathbf{z}}=\sum_{y\in\mathcal{Y}}\theta_{xy\mathbf{z}}\) for \(x\in\mathcal{X}\), \(\mathbf{z}\in\mathcal{Z}\).

Without further assumptions, drawing \(n\) independent observations of \((X,Y,\mathbf{Z})^T\) can be viewed as sampling from a multinomial distribution with parameters \(n\) and
\[\bm{\theta}=\{\theta_{xy\mathbf{z}}:x\in\mathcal{X},y\in\mathcal{Y},\mathbf{z}\in\mathcal{Z}\}.\]
We refer to this as the saturated multinomial model. The number of elements of \(\bm{\theta}\) is equal to \(|\mathcal{X}|\cdot|\mathcal{Y}|\cdot|\mathcal{Z}|\). As the elements must sum to 1, this corresponds to \(d=|\mathcal{X}|\cdot|\mathcal{Y}|\cdot|\mathcal{Z}|-1\) degrees of freedom.

If \(X\ind Y\mid \mathbf{Z}\), which is the null hypothesis of the \(G^2\)-test, then fewer parameters are required to describe the distribution of \((X,Y,\mathbf{Z})^T\). In particular, under \(X\ind Y\mid \mathbf{Z}\) we have that for all \(x\in\mathcal{X}\), \(y\in\mathcal{Y}\) and \(\mathbf{z}\in\mathcal{Z}\), \(\theta_{xy\mathbf{z}}=\theta_{+y\mathbf{z}}\cdot\theta_{x+\mathbf{z}}\). Thus, under the null hypothesis the set of parameters can be reduced to
\[\bm{\theta}^0=\{\theta_{+y\mathbf{z}}:y\in\mathcal{Y},\mathbf{z}\in\mathcal{Z}\}\cup\{\theta_{x+\mathbf{z}}:x\in\mathcal{X},\mathbf{z}\in\mathcal{Z}\},\]
which has \(|\mathcal{X}|\cdot|\mathcal{Z}|+|\mathcal{Y}|\cdot|\mathcal{Z}|\) elements. As \(\sum_{y{\in\mathcal{Y}}}\sum_{\mathbf{z}{\in\mathcal{Z}}}\theta_{+y\mathbf{z}}=\sum_{x{\in\mathcal{X}}}\sum_{\mathbf{z}{\in\mathcal{Z}}}\theta_{x+\mathbf{z}}=1\), this corresponds to \(d^0=(|\mathcal{X}|-1)\cdot|\mathcal{Z}|+(|\mathcal{Y}|-1)\cdot|\mathcal{Z}|\) degrees of freedom.

The \(G^2\)-test is a likelihood ratio test with test statistic
\[G^2=-2\left[l\left(\hat{\bm{\theta}}^0\right)-l\left(\hat{\bm{\theta}}\right)\right],\]
where \(l(\cdot)\) denotes the log-likelihood and the parameter estimates in \(\hat{\bm{\theta}}^0\) and \(\hat{\bm{\theta}}\) are obtained from the sample by counting the number of observations in the corresponding cell or margin and dividing by \(n\). Asymptotically and under the null hypothesis, \(G^2\) follows a \(\chi^2\)-distribution with \(d-d^0=(|\mathcal{X}|-1)\cdot(|\mathcal{Y}|-1)\cdot|\mathcal{Z}|\) degrees of freedom.

\subsection*{The CG-test}

The CG-distribution is defined as follows: Consider a set of variables \(\mathbf{V}\) partitioned into continuous variables \(\mathbf{C}\) and discrete variables \(\mathbf{B}\), where \(\mathbf{B}\) can take values in \(\mathcal{B}\). Then \(\mathbf{V}\) is said to follow a CG-distribution if for every \(\mathbf{b}\in\mathcal{B}\), the conditional distribution of \(\mathbf{C}\) given \(\mathbf{B}=\mathbf{b}\) is multivariate normal with mean vector \(\bm{\mu}_\mathbf{b}\) and covariance matrix \(\bm{\Sigma}_\mathbf{b}\) \citep{LauritzenWermuth1989, Lauritzen1990}. Note that \(\bm{\Sigma}_\mathbf{b}\) is allowed to depend on \(\mathbf{b}\), which is in contrast to the \textit{general location model} sometimes considered in the context of multiple imputation (\citealp{Schafer1997}, p.\,335). We denote the set of parameters describing the distribution of \(\mathbf{V}\) as \(\bm{\psi}_\mathbf{V}=\{p_\mathbf{b}, \bm{\mu}_\mathbf{b},\bm{\Sigma}_\mathbf{b} : \mathbf{b}\in\mathcal{B}\}\), where \(p_\mathbf{b}=\mathrm{P}(\mathbf{B}=\mathbf{b})\). The family of CG-distributions is not closed under marginalisation, i.e.\ if \(\mathbf{V}\) follows a CG-distribution, then a subset \(\mathbf{V}'\subset\mathbf{V}\) does not in general follow a CG-distribution (\citealp{Lauritzen1990}, Section 6.1.1).

A likelihood ratio test for conditional independence between CG-distributed variables was proposed by \cite{AndrewsRamseyCooper2018}. We call this the CG-test. Consider a random vector \((X,Y,Z_1,\dots,Z_s)^T\) following a CG-distribution with parameter vector \(\bm{\psi}_{XY\mathbf{Z}}\), where \(\mathbf{Z}=(Z_1,\dots,Z_s)\). For the CG-test, it is assumed that the marginal distributions of \((X,\mathbf{Z})\), \((Y,\mathbf{Z})\) and \(\mathbf{Z}\) are well approximated by CG-distributions with parameters \(\bm{\psi}_{X\mathbf{Z}}\), \(\bm{\psi}_{Y\mathbf{Z}}\) and \(\bm{\psi}_{\mathbf{Z}}\), respectively. As noted above, this does not in general follow from the assumption that \((X,Y,Z_1,\dots,Z_s)^T\) is CG.

Denote by \(\hat{\bm{\psi}}_{XY\mathbf{Z}}\), \(\hat{\bm{\psi}}_{X\mathbf{Z}}\), \(\hat{\bm{\psi}}_{Y\mathbf{Z}}\) and \(\hat{\bm{\psi}}_{\mathbf{Z}}\) the maximum likelihood estimates of  \(\bm{\psi}_{XY\mathbf{Z}}\), \(\bm{\psi}_{X\mathbf{Z}}\), \(\bm{\psi}_{Y\mathbf{Z}}\) and \(\bm{\psi}_{\mathbf{Z}}\), respectively, obtained from data, with corresponding log likelihoods \(l(\hat{\bm{\psi}}_{XY\mathbf{Z}})\), \(l(\hat{\bm{\psi}}_{X\mathbf{Z}})\), \(l(\hat{\bm{\psi}}_{Y\mathbf{Z}})\) and \(l(\hat{\bm{\psi}}_{\mathbf{Z}})\). The CG-test compares the log likelihood \(L=l(\bm{\psi}_{XY\mathbf{Z}}) / l(\bm{\psi}_{Y\mathbf{Z}})\) for modelling \(X\) given \(Y\) and \(\mathbf{Z}\) with the log likelihood \(L^0=l(\bm{\psi}_{X\mathbf{Z}}) / l(\bm{\psi}_{\mathbf{Z}})\) for modelling \(X\) given \(Y\) only, which corresponds to the null hypothesis that \(X\ind Y\mid\mathbf{Z}\), or equivalently, \(f(x\mid y,\mathbf{z})=f(x\mid \mathbf{z})\).

The test statistic of the CG-test is
\[\chi^2=-2(L^0-L).\]

Under the null hypothesis, \(\chi^2\) approximately follows a \(\chi^2\)-distribution. The degrees of freedom vary depending on which variables are continuous and which are discrete; for details see \cite{AndrewsRamseyCooper2018}\footnote{Note that equation (11) of \cite{AndrewsRamseyCooper2018} should read \(df_p(\hat{\theta}_p)=d(d+1)/2+1+\mathbf{d}\), in order to account for the estimated vector of means (Bryan Andrews, personal communication).}.

\subsection*{\(G^2\)-test and CG-test under multiple imputation}

Rules for combining likelihood ratio statistics have been suggested by \cite{MengRubin1992}. Consider \(M\) completed datasets obtained by multiple imputation. Let \(\bm{\phi}\) and \(\bm{\phi}^0\) be sets of parameters characterising the full and reduced model of interest. For the \(G^2\)-test, \(\bm{\phi}=\bm{\theta}\) and \(\bm{\phi}^0=\bm{\theta}^0\); for the CG-test, \(\bm{\phi}=\bm{\psi}\) and \(\bm{\phi}^0=\bm{\psi}^0\). As before, we use the superscript \(^{(m)}\) to indicate estimators obtained from the \(m\)-th completed dataset. We denote by \(l_m(\cdot)\) the log likelihood function given the \(m\)-th completed dataset.

First, the average likelihood ratio statistic is calculated as
\[\bar{LR}=\frac{1}{M}\sum_{m=1}^M -2[l_m(\bm{\hat{\phi}}^{0(m)})-l_m(\bm{\hat{\phi}}^{(m)})],\]
and the average parameter estimates as
\[\bm{\bar{\phi}}^0=\frac{1}{M}\sum_{m=1}^M \bm{\phi}^{0(m)}\]
and
\[\bm{\bar{\phi}}=\frac{1}{M}\sum_{m=1}^M \bm{\phi}^{(m)}.\]
The log likelihoods are then re-evaluated given each of the \(M\) completed datasets, with the parameters fixed to the average parameter estimates, and the corresponding likelihood ratio statistics are averaged:
\[\tilde{LR}=\frac{1}{M}\sum_{m=1}^M -2[l_m(\bm{\bar{\phi}}^0)-l_m(\bm{\bar{\phi}})].\]
The pooled test statistic is
\[D_3=\frac{\tilde{LR}}{k(1+r_3)}\]
with \(r_3=(M+1)(\bar{LR}-\tilde{LR})/[k(M-1)]\), where \(k\) 
equals the degrees of freedom that would have been used had complete data been available. The test statistic \(D_3\) can be approximated by an \(F\)-distribution with \(k\) and \(4+[k(M-1)-4][1+(1-2k^{-1}(M-1)^{-1})/r_3]^2\) degrees of freedom. The name `D\textsubscript{3}' has no specific meaning; it is used in several popular books to distinguish it from the so-called D\textsubscript{1} statistic for multi-parameter Wald tests and the so-called D\textsubscript{2} statistic for general \(\chi^2\)-tests \citep{Schafer1997, Enders2010, vanBuuren2018}.

\section{Identifiability of conditional (in)dependencies under test-wise deletion}
\label{app:proofs}

The following lemma on the identifiability of conditional dependencies is a rephrased version of Proposition~1 in \cite{Tuetal2019}:

\begin{lemma}
	\label{lemma:Tu}
Let \(\mathcal{D}\) be a missingness DAG with node set \(\mathbf{V}\cup\mathbf{R}(\mathbf{V})\), such that the distribution of \(\mathbf{V}\cup\mathbf{R}(\mathbf{V})\) is faithful to \(\mathcal{D}\), and assume that faithful observability holds. Let \(X,Y\in\mathbf{V}\) with \(X\ne Y\), and let \(\mathbf{Z}\subseteq\mathbf{V}\setminus\{X,Y\}\) such that \(X\nind Y\mid\mathbf{Z}\). Then \(X\nind Y\mid(\mathbf{Z},R^{XY\mathbf{Z}}=1)\).
\end{lemma}

In other words, under faithfulness and faithful observability, conditional dependencies are always identified under test-wise deletion. \cite{Tuetal2019} also show that conditional independencies are not always identified under test-wise deletion. Our next proposition provides a necessary and sufficient criterion for this type of identification. The proof builds on Theorem 6 of \cite{DidelezKreinerKeiding2010} and is based on the following properties of distributions faithful to DAGs (\citealp{Pearl1988}, Theorem 11):

Let \(\mathbf{A}\), \(\mathbf{B}\), \(\mathbf{C}\) and \(\mathbf{D}\) be disjoint subsets of a set of random variables \(\mathbf{V}\) faithful to a DAG with node set \(\mathbf{V}\).

Contraction: If \(\mathbf{A}\ind \mathbf{B}\mid \mathbf{C}\) and \(\mathbf{A}\ind \mathbf{D}\mid (\mathbf{B},\mathbf{C})\), then \(\mathbf{A}\ind(\mathbf{B},\mathbf{D}) \mid \mathbf{C}\).

Weak union: If \(\mathbf{A}\ind (\mathbf{B},\mathbf{D})\mid \mathbf{C}\), then \(\mathbf{A}\ind\mathbf{B}\mid (\mathbf{C},\mathbf{D})\).

Weak transitivity: If \(\mathbf{A}\ind\mathbf{B}\mid (\mathbf{C},D)\) and \(\mathbf{A}\ind\mathbf{B}\mid \mathbf{C}\), then either \(D\ind\mathbf{A}\mid\mathbf{C}\) or \(D\ind\mathbf{B}\mid\mathbf{C}\). Here \(D\) is required to be a singleton.

\begin{lemma}
	\label{lemma:ind}
	Let \(\mathbf{V}\) be a set of random variables with a joint distribution satisfying the properties of contraction, weak union and weak transitivity, and assume that faithful observability holds. Let \(X,Y\in\mathbf{V}\) with \(X\ne Y\), and let \(\mathbf{Z}\subseteq\mathbf{V}\setminus\{X,Y\}\) such that \(X\ind Y\mid\mathbf{Z}\). Then \(X\ind Y\mid(\mathbf{Z},R^{XY\mathbf{Z}}=1)\) if and only if \(R^{XY\mathbf{Z}}\ind X\mid(Y,\mathbf{Z})\) or \(R^{XY\mathbf{Z}}\ind Y\mid(X,\mathbf{Z})\).
\end{lemma}
\begin{proof}
	By faithful observability, \(X\ind Y\mid(\mathbf{Z},R^{XY\mathbf{Z}}=1)\Leftrightarrow X\ind Y\mid(\mathbf{Z},R^{XY\mathbf{Z}})\). We show that (i) \(X\ind Y\mid(\mathbf{Z},R^{XY\mathbf{Z}})\) if and only if (ii) \(R^{XY\mathbf{Z}}\ind X\mid(Y,\mathbf{Z})\) or \(R^{XY\mathbf{Z}}\ind Y\mid(X,\mathbf{Z})\).
	
	Suppose first that (i) holds. Then by weak transitivity, we have that either \(R^{XY\mathbf{Z}}\ind X\mid\mathbf{Z}\) or \(R^{XY\mathbf{Z}}\ind Y\mid\mathbf{Z}\). If \(R^{XY\mathbf{Z}}\ind X\mid\mathbf{Z}\), then by contraction, \((R^{XY\mathbf{Z}},Y)\ind X\mid\mathbf{Z}\), and by weak union, \(R^{XY\mathbf{Z}}\ind X\mid(Y,\mathbf{Z})\). Analogously, if \(R^{XY\mathbf{Z}}\ind Y\mid\mathbf{Z}\), then \(R^{XY\mathbf{Z}}\ind Y\mid(X,\mathbf{Z})\). Hence, (ii) holds.
	
	Suppose now that \(R^{XY\mathbf{Z}}\ind X\mid(Y,\mathbf{Z})\) holds. Since \(X\ind Y\mid\mathbf{Z}\), by contraction, \(X\ind (Y,R^{XY\mathbf{Z}})\mid \mathbf{Z}\). By weak union, \(X\ind Y\mid(\mathbf{Z},R^{XY\mathbf{Z}})\). By symmetry, if we instead suppose that \(R^{XY\mathbf{Z}}\ind Y\mid(X,\mathbf{Z})\), then \(Y\ind X\mid(\mathbf{Z},R^{XY\mathbf{Z}})\Leftrightarrow X\ind Y\mid(\mathbf{Z},R^{XY\mathbf{Z}})\), which completes the proof.
\end{proof}

We are now ready to prove Proposition~\ref{prop:twd} from Section~\ref{sec:twd}. For simplicity, it is assumed in the proof that oracle test-wise-deletion PC is based on the original version of oracle PC. However, the proposition also holds if the stable version proposed by \cite{ColomboMaathuis2014}, in which additional conditional independencies are considered, is used. See \cite{ColomboMaathuis2014} for pseudo-code for both variants.

\setcounter{theorem}{1}
\begin{proposition}
		Let \(\mathcal{D}\) be a missingness DAG with node set \(\mathbf{V}\cup\mathbf{R}(\mathbf{V})\), such that the distribution of \(\mathbf{V}\cup\mathbf{R}(\mathbf{V})\) is faithful to \(\mathcal{D}\), and assume that faithful observability holds. Then oracle test-wise-deletion PC recovers the true CPDAG over \(\mathbf{V}\) if and only if the admissible separator condition holds. 
\end{proposition}

\begin{proof}
	We first show that the skeleton part of oracle test-wise-deletion PC recovers the true skeleton if and only if the admissible separator condition holds.
	
	Suppose first that the skeleton is correctly recovered by oracle test-wise-deletion PC. This implies that for all pairs \((X,Y)\) of non-adjacent nodes in \(\mathbf{V}\), there exists a (possibly empty) set \(\mathbf{Z}\subseteq\mathrm{adj}(X,\mathcal{D})\) or  \(\mathbf{Z}\subseteq\mathrm{adj}(Y,\mathcal{D})\) such that \(X\ind Y\mid(\mathbf{Z},R^{XY\mathbf{Z}}=1)\), as otherwise the edge between \(X\) and \(Y\) would not have been removed during the algorithm. By Lemma~\ref{lemma:Tu}, \(X\ind Y\mid\mathbf{Z}\). By Lemma~\ref{lemma:ind}, \(X\ind Y\mid(\mathbf{Z},R^{XY\mathbf{Z}}=1)\) and \(X\ind Y\mid\mathbf{Z}\) together imply \(R^{XY\mathbf{Z}}\ind X\mid(Y,\mathbf{Z})\) or \(R^{XY\mathbf{Z}}\ind Y\mid(X,\mathbf{Z})\), hence the admissible separator condition is satisfied.
	
	Suppose now that the admissible separator condition is satisfied. Pick a pair \((X,Y)\) of non-adjacent nodes in \(\mathbf{V}\) and a set \(\mathbf{Z}\) satisfying the admissible separator condition with respect to \((X,Y)\), implying \(\mathbf{Z}\subseteq\mathrm{adj}(X,\mathcal{D})\) or \(\mathbf{Z}\subseteq\mathrm{adj}(Y,\mathcal{D})\) and \(R^{XY\mathbf{Z}}\ind X\mid(Y,\mathbf{Z})\) or \(R^{XY\mathbf{Z}}\ind Y\mid(X,\mathbf{Z})\). Then by Lemma~\ref{lemma:ind}, \(X\ind Y\mid(\mathbf{Z},R^{XY\mathbf{Z}}=1)\). Lemma~\ref{lemma:Tu} implies that no edges are erroneously deleted by oracle test-wise-deletion PC, i.e.\ the nodes adjacent to \(X\) in any intermediate graph obtained while the algorithm runs is a superset of \(\mathrm{adj}(X,\mathcal{D})\), and analogous for \(Y\). Hence, \(X\ind Y\mid(\mathbf{Z},R^{XY\mathbf{Z}}=1)\) is among the conditional independencies tested during the course of the algorithm. It follows that the skeleton is correctly recovered.
	
	The adjacencies are not further modified after the skeleton phase is completed. Hence, the necessary condition for the recovery of the true CPDAG is that the admissible separator condition holds. This proofs the `only if' direction of the statement in the proposition. For the other direction, suppose that the admissible separator condition holds for the remainder of the proof.
	
	Consider the v-structure phase of oracle test-wise-deletion PC. This phase is based on checking, for triples \((X,Y,Z)\) such that \(X-Y-Z\) is in the estimated skeleton and \(X-Z\) is not, whether \(Y\) is in the separating set \(\mathbf{W}\) conditionally on which \(X\) and \(Z\) were found to be independent in the skeleton phase. If \(Y\not\in\mathbf{W}\), then \(X-Y-Z\) is oriented as \(X\rightarrow Y\leftarrow Z\). We have already established above that under the admissible separator condition, \(X\ind Z\mid(\mathbf{W},R^{XZ\mathbf{W}})\Leftrightarrow X\ind Z\mid\mathbf{W}\). Hence, as we assume faithfulness, \(Y\not\in\mathbf{W}\) if and only if the true structure is \(X\rightarrow Y\leftarrow Z\). It follows that the v-structures are correctly recovered by oracle test-wise-deletion PC under the admissible separator condition.
	
	Finally, the orientation of additional edges is based on logical rules and returns the correct CPDAG as long as the skeleton and the v-structures have correctly been recovered.
\end{proof}

\FloatBarrier

\bibliography{MissBib2}

\newpage

\end{document}